\documentclass[12pt]{article}

\usepackage{color} 

\usepackage{amsfonts}
\usepackage{amsmath} 
\usepackage{amssymb}
\usepackage{amsthm}
\usepackage{graphicx}

\usepackage{hyperref}
\usepackage[dvipsnames]{xcolor}

\newtheorem{theorem}{Theorem}[section]                                          
                          
\newtheorem{lemma}[theorem]{Lemma}

\newtheorem{condition}{Hypothesis}

\newtheorem{definition}{Definition}[section]
\newtheorem{remark}{Remark}[section]

\usepackage{amsmath} 
\usepackage{amssymb}

\newcommand{\Tr}{\hbox{\rm tr}}

\begin{document}

\title{Basic properties of a mean field laser equation}

\author{ F. Fagnola  
\thanks{Dipartimento di Matematica, Politecnico di Milano,
Piazza Leonardo Da Vinci 32, I-20133, Milano, Italy.  e-mail: franco.fagnola@polimi.it } 
\and 
C.M. Mora 
\thanks{Departamento de Ingenier\'{\i}a Matem\'{a}tica, Universidad de Concepci\'on,
Barrio Universitario, Avenida   Esteban Iturra s/n, 4089100 , Casilla 160-C, Concepci\'on, Chile.
e-mail: cmora@ing-mat.udec.cl}
\thanks{Partially supported by the Universidad de Concepci\'on project VRID-Enlace 218.013.043-1.0.}
}

\date{ }

\maketitle

\abstract{
 We study the non-linear quantum master equation describing a laser under the mean field approximation. 
The quantum system is formed by a single mode optical cavity and two level atoms,
which interact with reservoirs.
Namely, we establish the existence and uniqueness of the regular solution to the non-linear operator
equation under consideration, as well as we get a probabilistic representation for this
solution in terms of a mean field stochastic Schr\"ondiger equation.
To this end,
we find a regular solution for the non-autonomous linear quantum master equation 
in Gorini-Kossakowski-Sudarshan-Lindblad form,
and 
we prove the uniqueness of the solution to the non-autonomous linear adjoint quantum master equation in Gorini-Kossakowski-Sudarshan-Lindblad form.
Moreover,
we obtain rigorously the Maxwell-Bloch equations from the mean field laser equation.
}


\section{Introduction}
 

This paper provides the mathematical foundation for the nonlinear laser equation
 \begin{eqnarray}
 \label{eq:Laser1}
 \frac{d }{dt} \rho_t
 \hspace{-7pt} & =  \hspace{-7pt}&
 -\mathrm{i} \frac{\omega}{2} \left[ 2 \, a^\dagger a + \sigma^{3} ,  \rho_t \right] 
 \\
 \nonumber
 & &
+  g \left[ 
\Bigl(  \Tr \left( \sigma^{-}  \rho_t \right) a^{\dagger}  -  \Tr \left( \sigma^{+}  \rho_t  \right) a \Bigr)
+
\Bigl(  \Tr\left( a^{\dagger}  \rho_t  \right) \sigma^{-}  -  \Tr\left( a \,  \rho_t  \right) \sigma^{+} \Bigr)         
 ,  \rho_t \right] 
\\
\nonumber
&&
+  \kappa_- \left( \sigma^{-}  \rho_t \,\sigma^{+} 
-\frac{1}{2} \sigma^{+} \sigma^{-}  \rho_t 
-\frac{1}{2}  \rho_t \,\sigma^+\sigma^{-}\right) 
\\
\nonumber
& &
+  \kappa_+ \left( \sigma^+  \rho_t \,\sigma^{-} 
-\frac{1}{2} \sigma^{-}\sigma^+  \rho_t
-\frac{1}{2} \ \rho_t \,\sigma^{-}\sigma^+\right)
\\ 
\nonumber
& &
+  2 \kappa \left( a\,  \rho_t a^\dagger - \frac{1}{2} 
a^\dagger a  \rho_t - \frac{1}{2}  \rho_t a^\dagger a\right) ,
\end{eqnarray}
where 
$\omega \in \mathbb{R}$, $g$ is a non-zero real number,   
$\kappa,\kappa_+,\kappa_- > 0$
and
$\rho_t$ is an unknown non-negative trace-class operator on $\ell^2 \left(\mathbb{Z}_+ \right) \otimes \mathbb{C}^2 $.
As usual, 
$ \left[  \cdot, \cdot \right] $ stands for the commutator of two operators,
$$
\sigma^{+} =
\left( \begin{array}{cc}
 0 & 1 \\ 0 & 0
\end{array} \right) ,
\quad
\sigma^{-} =
\left( \begin{array}{cc}
 0 & 0 \\ 1 & 0
\end{array} \right) ,
\quad
\sigma^{3} =
\left( \begin{array}{cc}
 1 & 0 \\ 0 & -1
\end{array} \right) ,
$$
and $a$, $a^{\dagger}$
are the closed operators  on $\ell^2 \left(\mathbb{Z}_+ \right)$ given by
$$
ae_{n} 
=
\left\{
\begin{array}{ll}  
\sqrt{n} \, e_{n-1}  &  \text{ if }   n \in  \mathbb{N} 
 \\ 0  &  \text{ if } n = 0  
\end{array} \right.
$$
and 
$
a^{\dagger}e_{n} = \sqrt{ n+1} \, e_{n+1} 
$
for all 
$ n \in   \mathbb{Z}_{+}$.
Here and subsequently,
$(e_n)_{n\ge 0}$ denotes the canonical orthonormal basis of $\ell^2(\mathbb{Z_+})$.

Under the mean field approximation,
(\ref{eq:Laser1}) describes 
the dynamics of a laser consisting of a radiation field coupled to a set of identical non-interacting two-level systems 
(see, e.g., Section 3.7.3 of \cite{BreuerPetruccione} 
and  
\cite{HeppLieb1974,MerkliBerman2012,Mori2013,Spohn1980} for more details on mean field quantum master equations).
The first term of the right-hand side of (\ref{eq:Laser1}) 
is determined by the free Hamiltonians of the field mode and the atoms,
the second term governs the atom-field interaction,
and the  last three terms,
i.e., the Gorini-Kossakowski-Sudarshan-Lindblad superoperators \cite{Gorini1976,Lindblad1976},
represent decay/pumping in the atoms and radiation losses.
We are interested in establishing rigorously the well-posedness of (\ref{eq:Laser1}), 
the equations of motion of the observables 
$a+ a^{\dagger}$, $\sigma^{-} + \sigma^{+}$ and $\sigma^{3}$,
and a probabilistic representation of $\rho_t$. 
This gives the mathematical basis to study, for instance, 
dynamical properties of (\ref{eq:Laser1}) and the numerical solution of (\ref{eq:Laser1}).

Our approach to the non-linear quantum master equation (\ref{eq:Laser1}) involves the study of  non-autonomous linear quantum master equations in the Gorini-Kossakowski-Sudarshan-Lindblad (GKSL) form \cite{AlickiLendi2007,BreuerPetruccione,Gorini1976,Lindblad1976}.
In the time-homogeneous setup, 
E. B. Davies and A. M. Chebotarev \cite{Chebotarev1991,Davies1977} constructed the minimal solution 
of GKSL linear master equations with unbounded coefficients (see, e.g., \cite{Chebotarev2000,Fagnola}).
Using semigroup methods, 
\cite{ChebGarQue98,ChebFagn1,ChebFagn2,Fagnola} prove that these equations have a  unique solution
under a quantum version of the Lyapunov condition for nonexplosion of classical Markov processes.
Applying probabilistic techniques,
one deduces that 
the GKSL quantum master equation preserves the regularity of the initial state (see, e.g., \cite{MoraAP} ),
and one also obtains the well-posedness of the GKSL adjoint quantum master equation with an initial condition given by an unbounded operator (see, e.g., \cite{MoraJFA}).
Using a limit procedure, 
one gets a conservative solution to a linear adjoint quantum master equation with time-dependent coefficients
(see, e.g, \cite{ChebGarQue97}).
In this article,
we address a class of time-local linear master equations,
which describes relevant physical systems (see, e.g., \cite{Breuer2016,BylickaChruscinskiManiscalco2014,ChruscinskiManiscalco2014,Hall2014,Schulte-Herbruggen2017}).
Namely, 
by extending some results given by \cite{MoraJFA,MoraAP},
we construct a regular solution for the non-autonomous linear quantum master equation 
\begin{equation}
\label{eq:3.50}
 \frac{d}{dt} \rho_{t}
 = 
 G \left( t \right) \rho_{t}
 +  \rho_{t}   G\left( t \right)^{\ast} 
 + \sum_{k=1}^{\infty}L_{k} \left( t \right) \rho_{t}  L_{k}\left( t \right)^{\ast}
 \hspace{1cm}
 t \geq 0  ,
\end{equation}
where 
$ \rho_{t}$ is a density operator in $\mathfrak{h}$, 
the initial datum $ \rho_{0}$ is regular,
and 
$G \left( t \right), L_{1} \left( t \right)$, $L_{2} \left( t \right), \ldots$ are linear operators in $\mathfrak{h}$
satisfying (on appropriate domain)
$$
G\left( t \right)=-  \mathrm{i} H\left( t \right)-\frac{1}{2} \sum_{\ell=1}^{\infty}L_{\ell}\left( t \right)^{\ast}L_{\ell} \left( t \right)
$$
with $ H\left( t \right) $ self-adjoint operator in $\mathfrak{h}$.
Furthermore,
we prove the uniqueness of the solution to the adjoint version of (\ref{eq:3.50}),
which models the evolution of the quantum observables in the Heisenberg picture.
This leads to prove the well-posedness of  the GKSL  quantum master equation  
resulting from replacing in (\ref{eq:Laser1})
the unknown values of
$\Tr\left( \sigma^{-}  \rho_t  \right) $ and $\Tr\left( a \,  \rho_t  \right)$
by known functions $\alpha \left( t \right) $ and $  \beta \left( t \right)$.

Our main objective is to develop the mathematical theory for the non-linear equation (\ref{eq:Laser1}).
First, 
we establish the existence and uniqueness of the regular solution to (\ref{eq:Laser1}).
In this direction,
Belavkin \cite{Belavkin1988,Belavkin1989} treated a general class of 
non-linear quantum master equations with bounded coefficients,
and
Kolokoltsov \cite{Kolokoltsov2010} obtained the well-posedness of nonlinear quantum dynamic semigroups
having 
non-linear Hamiltonians that are bounded perturbations of unbounded linear  self-adjoint operators,
together with 
non-linear bounded Gorini-Kossakowski-Sudarshan-Lindblad superoperators.
Arnold and Sparber \cite{ArnoldSparber2004} showed the existence and uniqueness of global solution to 
a non-linear quantum master equation involving Hartree potential by means of semigroup techniques.

Moreover, 
we deal with the equations of motion for the mean values of   
$a$, $ \sigma^{-}$ and  $ \sigma^{3} $.
It is well known that
the following  first-order differential equations is formally obtained from (\ref{eq:Laser1}):
 \begin{equation}
 \label{eq:laser2} 
\left\{ 
\begin{array}{lcl} 
  \frac{d}{dt} \Tr\left( a \, \rho_{t} \right) 
\hspace{-7pt} & = \hspace{-7pt}  &
 - \left( \kappa + \mathrm{i} \omega \right)  \Tr\left( a \, \rho_{t} \right) 
 + g \, \Tr\left( \sigma^{-} \rho_{t} \right)  
 \\
  \frac{d}{dt} \Tr\left( \sigma^{-} \rho_{t}  \right)
\hspace{-7pt} & = \hspace{-7pt}  &
 - \left( \gamma + \mathrm{i} \omega \right)   \Tr\left( \sigma^{-} \rho_{t} \right) 
 + g \ \Tr\left( a \, \rho_{t}  \right)  \Tr\left( \sigma^{3} \rho_{t}   \right) 
 \\
 \frac{d}{dt} \Tr\left(  \sigma^{3}  \rho_{t} \right)
\hspace{-7pt} & = \hspace{-7pt}  &
  - 4 g \ \Re \left(
  \Tr\left( a \, \rho_{t}  \right) \  \overline{ \Tr\left(  \sigma^{-}  \rho_{t}  \right) }
 \right)
 - 2 \gamma \left(  \Tr\left(  \sigma^{3} \rho_{t} \right) -d \right) 
\end{array}
\right.  ,
\end{equation}
where 
$ \geq 0$,
$\gamma = \left( \kappa_+ + \kappa_- \right)/2$
and
$
d = \left( \kappa_+ - \kappa_- \right) / \left( \kappa_+ + \kappa_- \right) 
$
(see, e.g., \cite{BreuerPetruccione}).
In the semiclassical laser theory,
the Maxwell-Bloch equations (\ref{eq:laser2}) describe the evolution of 
the field (i.e., $\Tr\left( a \, \rho_{t} \right)$), 
the polarization (i.e., $\Tr\left( \sigma^{-} \rho_{t}  \right)$) 
and the population inversion  (i.e., $ \Tr\left(  \sigma^{3}  \rho_{t} \right)$) 
of ring lasers like far-infrared $NH_3$ lasers  (see, e.g., \cite{Haken1985,Ohtsubo2013,TartwijkAgrawal1998}).
The  system (\ref{eq:laser2})
has received much attention in the physical literature 
due to its important role in the description of laser dynamics (see, e.g., \cite{BreuerPetruccione,Fowler1982,NingHaken1990}).
In this paper,  we prove rigorously the validity of (\ref{eq:laser2}) 
whenever the initial state is regular enough,
and thus we get an Ehrenfest theorem for (\ref{eq:Laser1})
(see, e.g., \cite{FagMora2013,Friesecke2009,Friesecke2010}).

Finally,
we obtain a probabilistic representation of (\ref{eq:Laser1}).
The solution of the linear quantum master equations in GKSL form
is characterized as the mean value of random pure states given by 
the linear and non-linear stochastic Schr\"odinger equations 
(see, e.g., \cite{Barchielli,BarchielliHolevo1995,BreuerPetruccione,MoraAP,WisemanMilburn2009}).
This representation plays an important tool in the numerical simulation of open quantum systems 
(see, e.g., \cite{BreuerPetruccione,MoraAAP2005,MoraFernBiscay2018,Percival,Schack1995}), 
and it has also been used for proving theoretical properties of  the GKSL quantum master equations  
(see, e.g., \cite{FagnolaMora2015,MoraJFA,MoraAP}).
In this paper,
we get  a probabilistic representation of (\ref{eq:Laser1})
in terms of a mean field version of the linear stochastic Schr\"odinger equation.
To the best of our knowledge this is the first rigorously established result, at the level of infinite dimensional density matrices, with an unbounded nonlinear evolution operator, in the study of nonlinear mean field laser  evolution equations

This paper is organized as follows.
Section \ref{sec:QuantumBifurcation} presents the main results.
Section \ref{sec:LinearQMEs} is devoted to general linear master equations.
In Section \ref{sec:AuxiliaryEquations}  we study 
a linear quantum master equation associated with (\ref{eq:Laser1}),
moreover, for the sake of completeness, we recall the basic properties of the complex Lorenz equations.
All proofs are deferred to Section \ref{sec:Proofs}.

\subsection{Notation}
\label{subsec:not}

In this paper,  
$\left(\mathfrak{h},\left\langle \cdot,\cdot\right\rangle \right) $ is a separable complex Hilbert space,
where the scalar product 
$\left\langle \cdot,\cdot \right\rangle $ is linear in the second variable and anti-linear in the first one. 
The standard basis of $ \mathbb{C}^2 $ is denoted by 
$
e_+ =
\begin{pmatrix}
 1 \\ 0
\end{pmatrix}
$
and 
$
e_- =
\begin{pmatrix}
 0 \\ 1
\end{pmatrix} 
$.
If $A,B$  are linear operators   in  $\mathfrak{h}$,
then 
$ \left[ A,  B \right] = AB - BA$
and 
$\mathcal{D}\left(A\right)$ stands for the domain of $A$.
We take $N= a^\dagger a $.
In case $\mathfrak{X}$, $\mathfrak{Z}$ are normed spaces,
we denote by $\mathfrak{L}\left( \mathfrak{X},\mathfrak{Z}\right) $ the set of all bounded operators from $\mathfrak{X}$ to $\mathfrak{Z}$
and we choose $\mathfrak{L}\left( \mathfrak{X}\right) = \mathfrak{L}\left( \mathfrak{X},\mathfrak{X}\right) $. 
We write $\mathfrak{L}_{1}\left( \mathfrak{h}\right)$ for the set of all trace-class operators on $\mathfrak{h}$ 
equipped with the trace norm.
For simplicity of notation,
generic no-negative constants are denoted by $K$, as well as 
$K \left( \cdot \right)$ stands for different  non-decreasing 
non-negative functions on $\left[ 0, \infty \right[$.

Let $C$ be a self-adjoint positive operator in $\mathfrak{h}$. 
Then,
$\pi _{C}:\mathfrak{h\rightarrow h}$  is defined by 
$\pi_C(x)=x$ if $x\in \mathcal{D}\left( C\right) $ and 
$\pi_C(x)=0$ if $x\notin \mathcal{D}\left( C\right) $,
as well as 
$ \left\Vert x\right\Vert _{C}=
\sqrt{\left\langle x,x\right\rangle _{C}}$
with 
$\left\langle x,y\right\rangle_{C}=\left\langle x,y\right\rangle +\left\langle Cx,Cy\right\rangle $
for  any $x,y\in \mathcal{D}\left( C\right) $.
We write $L^{2}\left( \mathbb{P},\mathfrak{h}\right) $ for the set of all square integrable random variables from $\left( \Omega ,\mathfrak{F},\mathbb{P}\right)$ to $ \left( \mathfrak{h},\mathfrak{B}\left( \mathfrak{h}\right) \right)$,
where  $ \mathcal{B} \left( \mathfrak{Y} \right)$ is the collection of all Borel set of the topological space $ \mathfrak{Y}$.
Finally,
$L_{C}^{2}\left( \mathbb{P},\mathfrak{h}\right) $  denotes  the set of all $\xi \in L^{2}\left( \mathbb{P},\mathfrak{h}\right) $ 
 satisfying $\xi \in \mathcal{D}\left( C\right) $ a.s. and $\mathbb{E} \left( \left\Vert \xi \right\Vert _{C}^{2} \right) <\infty $. 


\section{Basic properties of the mean field laser equation}
\label{sec:QuantumBifurcation}


This section presents the main results of the paper,
which are summarized in Theorem \ref{th:EyU-LaserE} given below.
We start by adapting the notion of regular weak  solution
---of a linear quantum master equation (see, e.g., \cite{MoraAP} and Definition \ref{def:RegularSolQME} given below)---
to the mean field laser equation (\ref{eq:Laser1}).
To this end,
we recall that a density operator $\varrho$ is $C$-regular if, roughly speaking, 
$C \varrho \, C$ is a trace-class operator,
where $C$ is a suitable reference operator (see, e.g., \cite{ChebGarQue98,MoraAP}).

\begin{definition}
\label{def2}
Suppose that $C$ is a self-adjoint positive operator in $\mathfrak{h}$. 
An operator $\varrho \in \mathfrak{L}_1\left( \mathfrak{h} \right)$
is called density operator iff  $\varrho$ is a non-negative operator with unit trace.
The non-negative operator  $\varrho \in \mathfrak{L}\left( \mathfrak{h} \right)$ is said to be $C$-regular 
iff 
$
 \varrho=\sum_{n\in\mathfrak{I}}\lambda_{n}\left\vert u_{n}\rangle\langle u_{n}\right\vert
$
for some countable set $\mathfrak{I}$, summable non-negative real numbers $\left( \lambda_{n}\right) _{n\in\mathfrak{I}}$ and collection $\left( u_{n}\right) _{n\in\mathfrak{I}}$ of elements of $\mathcal{D}\left( C\right) $, which together satisfy:
$
 \sum_{n\in \mathfrak{I}}\lambda_{n}\left\Vert Cu_{n}\right\Vert ^{2}<\infty
 $.
 Let $\mathfrak{L}_{1,C}^{+} \left( \mathfrak{h}\right) $ denote the set of all $C$-regular density operators in $\mathfrak{h}$. 
 \end{definition}

 \begin{definition}
Let $C$ be a self-adjoint positive operator in $\mathfrak{h}$.
A family $\left( \rho_t \right)_{t \geq 0}$ of operators belonging to $\mathfrak{L}_{1,C}^{+} \left( \mathfrak{h}\right) $
is called
$C$-weak solution to (\ref{eq:Laser1}) iff 
 the function $t \mapsto \Tr\left( a \rho_{t}  \right) $ is continuous
 and for all $t\geq0$ we have
 $$
 \frac{d}{dt}\Tr\left( A \rho_{t}  \right) 
 = 
\Tr\left( A  \mathcal{L}_{\star} \left( \rho_t \right)  \rho_t  \right) 
\qquad 
\forall A\in\mathfrak{L}\left( \mathfrak{h}\right) ,
$$
where 
\begin{eqnarray*}
 \mathcal{L}_{\star} \left( \widetilde{\varrho}   \right)  \varrho
  \hspace{-7pt} & =  \hspace{-7pt}& 
  -\frac{\mathrm{i} \omega}{2}
\left[ 2 a^\dagger a +\sigma^{3}, \varrho \right]   
+  2 \kappa \left( a\, \varrho a^\dagger - \frac{1}{2} 
a^\dagger a \varrho - \frac{1}{2}\varrho a^\dagger a\right) 
\\ 
& &
+   \kappa_-  \left( \sigma^{-} \varrho \,\sigma^+ 
-\frac{1}{2} \sigma^+\sigma^{-} \varrho
-\frac{1}{2} \varrho \,\sigma^+\sigma^{-}\right) 
\\ 
& &
+   \kappa_+  \left( \sigma^+ \varrho \,\sigma^{-} 
-\frac{1}{2} \sigma^{-}\sigma^+ \varrho
-\frac{1}{2} \varrho \,\sigma^{-}\sigma^+\right) 
\\  
& &
 +  g \left[  \Tr\left( \sigma^{-}  \widetilde{\varrho}  \right) a^\dagger
                   -  \Tr\left( \sigma^{+}  \widetilde{\varrho}  \right) a, \varrho \right] 
 +  g \left[ \Tr\left( a^\dagger  \widetilde{\varrho}  \right) \sigma^{-}
                   -  \Tr\left( a \,  \widetilde{\varrho}  \right) \sigma^+, \varrho \right] .
\end{eqnarray*}
\end{definition}

Similar to the linear case,
(\ref{eq:Laser1}) is strongly related with the following non-linear stochastic evolution equation on 
$\ell^2 \left(\mathbb{Z}_+ \right) \otimes \mathbb{C}^2$:
\begin{eqnarray}
\label{eq:SSENonlinear}
 Z_{t}\left( \xi \right) 
& = &
\xi 
+ \int_{0}^{t} \left( -  \mathrm{i} H \left( t , Z_{t} \left( \xi \right) \right) - \frac{1}{2} \sum_{\ell = 1}^3 L_{\ell}^* L_{\ell} \right) 
Z_{s}\left( \xi \right) ds 
\\
\nonumber
& &
+ \sum_{\ell=1}^{ 3 }\int_{0}^{t}
L_\ell \, Z_{s}\left( \xi \right) dW_{s}^{\ell} ,
\end{eqnarray}
where 
$L_1  =  \sqrt{2 \kappa} \, a $, 
$L_2 =  \sqrt{ \gamma \left(1-d \right) } \, \sigma^{-}$,
$L_3 =  \sqrt{ \gamma \left(1 +d \right) } \, \sigma^{+} $,
\begin{eqnarray}
\label{eq:HZt}
H \left( t ,  Z_t \left( \xi \right)\right) 
& = &
\frac{\omega}{2} \left( 2 a^\dagger a  +\sigma^{3} \right) 
\\ 
\nonumber
& &
+ \mathrm{i}  g \left(
\mathbb{E} \left\langle Z_{t} \left( \xi \right), \sigma^{-} Z_{t} \left( \xi \right) \right\rangle a^\dagger 
-  
\mathbb{E} \left\langle Z_{t} \left( \xi \right), \sigma^{+} Z_{t} \left( \xi \right) \right\rangle a 
\right)
\\
\nonumber
& &
+ \mathrm{i}  g 
\left( \mathbb{E} \left\langle Z_{t} \left( \xi \right), a^\dagger Z_{t} \left( \xi \right) \right\rangle  \sigma^{-}  
- 
\mathbb{E} \left\langle Z_{t} \left( \xi \right), a \, Z_{t} \left( \xi \right) \right\rangle \sigma^+ \right) ,
\end{eqnarray}
and
$W^1, W^2, W^3$ are real valued independent Wiener processes on a filtered complete probability 
space $\left( \Omega ,\mathfrak{F}, \left(\mathfrak{F}_{t}\right) _{t\geq 0},\mathbb{P}\right) $.
Next, we tailor the concept of regular weak solution
---used in the framework of stochastic Sch\"odinger equations 
(see, e.g., \cite{FagMora2013,MoraReIDAQP,MoraReAAP} and Definition \ref{def:regular-sol} given below)---
to suit (\ref{eq:SSENonlinear}).

\begin{definition}
Let $p \in \mathbb{N}$.
An $\ell^2 \left(\mathbb{Z}_+ \right) \otimes \mathbb{C}^2$-valued adapted process   
with continuous sample paths $\left( Z_{t}\left( \xi \right) \right) _{t \in \mathbb{I}}$
is called strong $N^p$-solution of (\ref{eq:SSENonlinear})
if:
\begin{itemize}
 \item    For all $t \geq 0$:
 $ 
\mathbb{E} \left\Vert Z_{t}\left( \xi \right)  \right\Vert ^{2}
\leq
K \left( t \right)  \mathbb{E}\left\Vert \xi \right\Vert ^{2}
$,
$
Z_{t}\left( \xi \right) \in \mathcal{D}\left( N^p \right) $ a.s.,
and
$$
\sup_{s\in \left[ 0,t\right] }\mathbb{E} \left\Vert N^p X_{s}\left( \xi \right)  \right\Vert ^{2} < \infty .
$$

\item The functions 
$
t \mapsto \mathbb{E} \left\langle Z_{t} \left( \xi \right), \sigma^{-} Z_{t} \left( \xi \right) \right\rangle
$
and 
$
t \mapsto \mathbb{E} \left\langle Z_{t} \left( \xi \right), a \, Z_{t} \left( \xi \right) \right\rangle
$
are continuous.
 
\item a.s. for all $t \geq 0$:
\begin{eqnarray*}
Z_{t}\left( \xi \right) 
& = &
\xi 
+ \int_{0}^{t} \left( -  \mathrm{i} H \left( t \right) - \frac{1}{2} \sum_{\ell = 1}^3 L_{\ell}^* L_{\ell} \right) 
\pi _{N^p}\left( Z_{s}\left( \xi \right) \right) ds 
\\
& &
+ \sum_{\ell=1}^{ 3 }\int_{0}^{t}
L_\ell \,\pi _{N^p}\left( Z_{s}\left( \xi \right) \right)  dW_{s}^{\ell} 
\end{eqnarray*}
with $H \left( t ,  Z_t \left( \xi \right)\right)$ described by (\ref{eq:HZt}), and 
$L_{\ell}$, $W^{\ell}$ as in (\ref{eq:SSENonlinear}).
\end{itemize}
\end{definition}

Now, 
we establish 
the existence and uniqueness of the regular solution to (\ref{eq:Laser1}),
a Ehrenfest-type theorem describing the evolution  of the mean values of the observables 
$a+ a^{\dagger}$, $\sigma^{-} + \sigma^{+}$ and $\sigma^{3}$,
and the probabilistic representation of (\ref{eq:Laser1}).

\begin{theorem}
 \label{th:EyU-LaserE}
Suppose that  
$\varrho \in \mathfrak{L}_{1,N^p}^{+} \left( \ell^2 \left(\mathbb{Z}_+ \right) \otimes \mathbb{C}^2 \right) $, with  $p \in \mathbb{N}$.
Then,
there exists a unique $N^p$-weak solution  $\left( \rho_t  \right)_{t \geq 0}$  to (\ref{eq:Laser1}) 
with initial datum $ \varrho$.
Moreover, the Maxwell-Bloch equations (\ref{eq:laser2}) hold,  and 
\begin{equation}
\label{eq:RepProbNolinear}
 \rho_t = \mathbb{E} \left| Z_{t} \left( \xi \right)\right\rangle \left\langle Z_{t} \left( \xi \right)\right| 
\hspace{2cm}
\forall t \geq 0 , 
\end{equation}
where  $\xi \in L_{N^p}^{2}\left( \mathbb{P},\mathfrak{h}\right) $ satisfies
$
\varrho = \mathbb{E} \left| \xi \right\rangle \left\langle \xi \right| 
$, 
and 
$Z_{t} \left( \xi \right) \in  \ell^2 \left(\mathbb{Z}_+ \right) \otimes \mathbb{C}^2 $ is the strong $N^p$-solution of (\ref{eq:SSENonlinear}).
\end{theorem}

\begin{proof}
 Deferred to Section  \ref{sec:Proof:EyU-LaserE}.
\end{proof}

\begin{remark}
 If $g^2 d  <   \kappa \gamma $, then
$\left( 0, 0, d \right)$ is an asymptotically stable equilibrium point of (\ref{eq:laser2}).
In fact, 
from  (\ref{eq:L5}) and (\ref{eq:L4}), given below, it follows  that 
$\Tr\left( a \, \rho_{t} \right)$, $\Tr\left( \sigma^{-} \rho_{t}  \right) $ and $\Tr\left(  \sigma^{3}  \rho_{t} \right) - d $
converge exponentially fast to $0$ as $t$ goes to $+ \infty$.
\end{remark}

\section{General linear quantum master equations}
\label{sec:LinearQMEs}

\subsection{Regular solution for the GKSL quantum master equation}

This subsection provides a regular solution for the linear quantum master equation (\ref{eq:3.50}). 
By generalizing  \cite{MoraAP} to the  non-autonomous framework,
we will describe a solution of (\ref{eq:3.50}) with the help of 
the linear stochastic evolution equation in $\mathfrak{h}$:
\begin{equation}
\label {eq:SSE}
X_{t}\left( \xi \right) 
= \xi +\int_{0}^{t}G \left( s \right) X_{s}\left( \xi \right) ds 
+ \sum_{\ell=1}^{\infty }\int_{0}^{t}
L_\ell \left( s \right) X_{s}\left( \xi \right) dW_{s}^{\ell} ,
\end{equation}
where
$W^1, W^2, \ldots$ are real valued independent Wiener processes on a filtered complete probability 
space $\left( \Omega ,\mathfrak{F}, \left(\mathfrak{F}_{t}\right) _{t\geq 0},\mathbb{P}\right) $.

Suppose that the density operator $\rho_0$ is $C$-regular.
According to Theorem 3.1 of \cite{MoraAP} we have
$
\rho_0 = \mathbb{E}\left\vert \xi \rangle \langle \xi \right\vert 
$
for certain $\xi \in L_{C}^{2}\left( \mathbb{P},\mathfrak{h}\right) $.
We set 
\begin{equation}
 \label{eq:def-rho-t}
  \rho_t := \mathbb{E} \left|X_{t} \left( \xi \right)\right\rangle \left\langle X_{t} \left( \xi \right)\right| ,
\end{equation}
where we use Dirac notation,
$X_{t} \left( \xi \right)$ is the unique strong $C$-solution of (\ref{eq:SSE})
(see Definition \ref{def:regular-sol}),
and the mathematical expectation 
can be interpreted as a Bochner integral in both $\mathfrak{L}_{1}\left( \mathfrak{h}\right)$ and $\mathfrak{L}\left( \mathfrak{h}\right)$.
Then, 
$\rho_t $ is a $C$-regular density operator (see  \cite{MoraAP} for details).

\begin{condition}
\label{hyp:L-G-C-domain}
There exists a self-adjoint positive operator $C$ in $\mathfrak{h}$
such that 
 $\mathcal{D}\left(C \right) \subset \mathcal{D}\left( G  \left( t \right) \right)$ 
 and 
 $\mathcal{D}\left(C \right) \subset \mathcal{D}\left( L_{\ell} \left( t \right) \right)$
 for all $t \geq 0$,
 and
 $G \left( \cdot \right) \circ \pi _{C}$ and $L_{\ell}  \left( \cdot \right)  \circ \pi _{C}$
are measurable as  functions from 
 $\left( \left[ 0 , \infty \right[ \times \mathfrak{h}, 
\mathcal{B}\left(  \left[ 0 , 
\infty \right[ \times \mathfrak{h} \right) \right) $ to 
$\left(\mathfrak{h}, \mathcal{B} \left( \mathfrak{h} \right)\right)$.
\end{condition}

\begin{definition}
 \label{def:regular-sol}
Assume Hypothesis \ref{hyp:L-G-C-domain}. 
Let $\mathbb{I}$ be
either $\left[ 0,\infty \right[ $ or $\left[ 0,T\right] $, 
with $T\in \mathbb{R}_{+}$. 
By strong $C$-solution of (\ref{eq:SSE}) with initial condition $\xi$, on the interval $\mathbb{I}$, 
we mean an $\mathfrak{h}$-valued adapted process $\left( X_{t}\left( \xi \right) \right) _{t \in \mathbb{I}}$  
with continuous sample paths such that for all $t\in \mathbb{I}$:
$ 
\mathbb{E} \left\Vert X_{t}\left( \xi \right)  \right\Vert ^{2}
\leq
K \left( t \right)  \mathbb{E}\left\Vert \xi \right\Vert ^{2}
$,
$X_{t}\left( \xi \right) \in \mathcal{D}\left( C\right) $ a.s., 
$
\sup_{s\in \left[ 0,t\right] }\mathbb{E} \left\Vert C X_{s}\left( \xi \right)  \right\Vert ^{2} < \infty 
$,
and 
$$
X_{t}\left( \xi \right) 
=\xi 
+\int_{0}^{t}G \left( s \right) \pi _{C}\left( X_{s}\left( \xi \right) \right) ds
+\sum_{\ell=1}^{\infty }
\int_{0}^{t}L_\ell \left( s \right) \pi _{C}\left( X_{s}\left( \xi \right) \right) dW_{s}^\ell
\hspace{0.3cm} a.s.
$$
\end{definition}

The following theorem,
which extends Theorem 4.4 of \cite{MoraAP} to the non-autonomous context,
asserts that $\rho_t$ given by (\ref{eq:def-rho-t})
is a regular solution to (\ref{eq:3.50}). 


\begin{definition} 
\label{def:RegularSolQME}
Let $C$ be a self-adjoint positive operator in $\mathfrak{h}$.
 A family $\left( \rho_t \right)_{t \geq 0}$ of $C$-regular  density operators is called
 $C$-weak solution to (\ref{eq:3.50}) 
 if and only if 
 \begin{equation}
 \label{3.12}
\frac{d}{dt}\Tr\left( A\rho_{t} \right) 
=
\Tr\left(
A\left( G \left( t \right) \rho_{t}\ +\rho_{t} G  \left( t \right)  ^{\ast}
+\sum_{\ell=1}^{\infty} L_{\ell}  \left( t \right)  \rho_{t} L_{\ell}  \left( t \right)  ^{\ast }\right) 
\right)
\end{equation}
 for all $A\in\mathfrak{L}\left( \mathfrak{h}\right) $ and $t\geq0$.
\end{definition}

\begin{condition}
\label{hyp:formal-conservativity} 
Suppose that $C$ satisfies Hypothesis \ref{hyp:L-G-C-domain}, together with:
\begin{itemize}
\item[(H2.1)] For any $t \geq 0$ and 
$x \in \mathcal{D}\left( C\right)$,
$\left\Vert  G  \left( t \right) x  \right\Vert^{2} \leq K \left( t \right) \left\Vert  x  \right\Vert_{C}^{2}$ .

\item[(H2.2)]  
For any  $t \geq 0$ and $x\in \mathcal{D}\left( C\right) $,
$
 2\Re\left\langle x, G \left( t \right) x\right\rangle 
+ \sum_{\ell=1}^{\infty }\left\Vert L_\ell \left( t \right) x\right\Vert ^{2} = 0.
$

\item[(H2.3)] 
For any initial datum $\xi \in L_{C}^{2}\left( \mathbb{P},\mathfrak{h}\right)$,
(\ref{eq:SSE}) has a unique strong $C$-solution on any bounded interval.

\item[(H2.4)] There exist functions
$
f_k : \left[ 0 , \infty \right[ \times \left[ 0 , \infty \right[  \rightarrow \left[ 0 , \infty \right[
$
such that:
(i) $f_k$ is bounded on bounded subintervals of $\left[ 0 , \infty \right[ \times  \left[ 0 , \infty \right[$;
(ii) 
$
\lim_{s \rightarrow t} f_k \left( s, t \right) = 0
$;
and 
(iii)
for all $s, t \geq 0$ and $x \in \mathcal{D}\left( C\right)$ we have
$
\left\Vert  G  \left( s \right) x - G  \left( t \right) x  \right\Vert^{2} \leq f_{0} \left( s, t \right) \left\Vert  x  \right\Vert_{C}^{2} 
$
and 
$
\left\Vert  L_{\ell}  \left( s \right) x - L_{\ell}  \left( t \right) x  \right\Vert^{2} \leq f_{\ell} \left( s, t \right) \left\Vert  x  \right\Vert_{C}^{2}
$.
\end{itemize}
\end{condition}

\begin{theorem}
\label{teor10}
Let Hypotheses \ref{hyp:L-G-C-domain} and \ref{hyp:formal-conservativity} hold.
Assume that $\varrho_0$ be $C$-regular, 
and that 
$G \left( t \right) , L_{1}  \left( t \right)$, $L_{2} \left( t \right)$, $\ldots $ are closable for all  $t \geq 0$.
Then $\rho_{t} $ given by (\ref{eq:def-rho-t}) is a $C$-weak solution to (\ref{eq:3.50}).
Moreover, for all $t \geq 0$ we have
\begin{equation}
\label{3.11}
\rho_{t} =\rho_0 + \int_{0}^{t}
\left( G \left( s \right) \rho_{s} +\rho_{s} G \left( s \right)^{\ast}
+\sum_{\ell=1}^{\infty } L_{\ell} \left( s \right) \rho_{s} L_{\ell}\left( s \right)^{\ast}\right) ds ,
\end{equation}
where 
we understand the above integral in the sense of Bochner integral in $\mathfrak{L}_{1}\left( \mathfrak{h}\right) $.
\end{theorem}

\begin{proof}
 Deferred to Section \ref{sec:ProofTh-teor10}.
\end{proof}
 
\begin{remark}
Sufficient conditions for the regularity of the solution to the linear stochastic Sch\"odinger equation (\ref{eq:SSE})
(i.e., Hypothesis 2.3)
are given, for instance, in \cite{FagMora2013,MoraMC2004,MoraReIDAQP}.
\end{remark}

\subsection{Uniqueness of the solution to the adjoint  quantum master equation in the GKSL form}

The next theorem introduces the operator $\mathcal{T}_{t}\left( A \right)$
that describes the evolution of the observable $A$ at time $t$ in the Heisenberg picture.
Roughly speaking,
the maps $ A \mapsto \mathcal{T}_{t}\left( A \right)$
is the adjoint operator of the application  $\varrho \mapsto \rho_t $,
where $\rho_t$ is defined by (\ref{eq:def-rho-t}).

\begin{condition}
\label{hyp:Well-posed} 
Let Hypothesis \ref{hyp:L-G-C-domain} hold,
together with Conditions H2.1 and H2.3.
Suppose that 
\begin{itemize}

\item[(H3.1)] For all  $t \geq 0$ and $x\in \mathcal{D}\left( C\right) $,
$$
 2\Re\left\langle x, G \left( t \right) x\right\rangle 
+ \sum_{\ell=1}^{\infty }\left\Vert L_\ell \left( t \right) x\right\Vert ^{2}  \leq K \left( t \right) \left\Vert  x  \right\Vert^{2}.
$$
\end{itemize}
\end{condition}

\begin{theorem}
 \label{th:SesqForm}
 Assume that Hypothesis \ref{hyp:L-G-C-domain} and Conditions H2.1 and H2.3 holds.
Consider $A \in \mathfrak{L}\left( \mathfrak{h}\right)$.   
Then, 
 for every $t \geq 0$ there exists a unique $\mathcal{T}_{t}\left( A \right) \in \mathfrak{L}\left( \mathfrak{h}\right)$ for which:
\begin{equation}
\label{4.2}
\left\langle x, \mathcal{T}_{t}\left( A \right) y \right\rangle = \mathbb{E} \left\langle X_{t} \left( x \right), A  X_{t} \left( y \right) \right\rangle 
\qquad \forall x,y \in \mathcal{D}\left(C  \right).
\end{equation}
Moreover, 
$ \sup_{t \in \left[ 0,T \right] } \left\Vert \mathcal{T}_{t}\left( A \right) \right\Vert < \infty$ for all $T \geq 0$.
\end{theorem}

\begin{proof}
 Deferred to Section  \ref{sec:ProofTh-SesqForm}.
\end{proof}

Theorem \ref{th:UniquenessLinearHE}  below shows  that  $\mathcal{T}_{t}\left( A \right)$ is the unique possible solution of the adjoint  quantum master equation
\begin{equation}
\label{eq:AdjointQME}
 \frac{d}{dt} \mathcal{T}_{t}\left( A \right)
 = 
 \mathcal{T}_{t}\left( A \right) G \left( t \right) 
 +  G\left( t \right)^{\ast}  \mathcal{T}_{t}\left( A \right)   
 + \sum_{k=1}^{\infty}L_{k} \left( t \right)^{\ast} \mathcal{T}_{t}\left( A \right)  L_{k}\left( t \right) .
 \end{equation}
Thus, 
we generalize Theorem 2.2 of \cite{MoraAP} to the non-autonomous framework.

\begin{theorem}
\label{th:UniquenessLinearHE}
Let  Hypothesis \ref{hyp:Well-posed} hold,
and let $\mathcal{T}_{t}\left( A \right)$ be as in Theorem \ref{th:SesqForm} with $A \in \mathfrak{L}\left( \mathfrak{h}\right)$.
Assume that $\left( \mathcal{A}_{t} \right)_{t \geq 0}$ is a family of operators  belonging to $\mathfrak{L}\left( \mathfrak{h}\right)$
such that 
$\mathcal{A}_{0} = A$, 
$\sup_{s \in \left[ 0,t\right]} \left\Vert \mathcal{A}_{s}  \right\Vert_{\mathfrak{L}\left( \mathfrak{h}\right)} < \infty$,
and 
\begin{equation}
\label{eq:4.1}
 \frac{d}{dt} \left\langle x, \mathcal{A}_{t} y \right\rangle 
= 
\left\langle x, \mathcal{A}_{t} G \left( t \right)  y \right\rangle + \left\langle G \left( t \right)  x, \mathcal{A}_{t} y \right\rangle  
+ \sum_{\ell =1}^{\infty} \left\langle L_{\ell} \left( t \right)  x, \mathcal{A}_{t} L_{\ell} \left( t \right)  y \right\rangle
\end{equation}
for all $x,y \in \mathcal{D}\left(C  \right)$.
Then  $\mathcal{A}_{t} = \mathcal{T}_{t}\left( A \right)$ for all $t \geq 0$.
\end{theorem}

\begin{proof}
Deferred to Section \ref{sec:ProofTh-UniquenessLinearHE}.
\end{proof}

\begin{remark}
In the autonomous case,
\cite{MoraJFA,MoraAP} obtain sufficient conditions for 
$\mathcal{T}_{t}\left( A \right)$ defined by (\ref{4.2}) to be solution of (\ref{eq:AdjointQME}).
Using semigroup methods,
\cite{ChebGarQue98,ChebFagn1,ChebFagn2,Fagnola} 
show the existence and uniqueness of solutions to (\ref{eq:3.50}) and (\ref{eq:AdjointQME}),
in the semigroup sense.
\end{remark}

In order to check Condition H2.3 we establish the following extension of Theorem 2.4 of \cite{FagMora2013}.

\begin{condition}\label{hyp:CF-inequality}
Suppose that $C$ satisfies Hypothesis \ref{hyp:L-G-C-domain},
together with:
\begin{itemize}

\item[(H4.1)]  For any $t \geq 0$ and 
$x \in \mathcal{D}\left( C\right)$,
$\left\Vert  G  \left( t \right) x  \right\Vert^{2} \leq K \left( t \right) \left\Vert  x  \right\Vert_{C}^{2}$ .

\item[(H4.2)] For every $\ell \in \mathbb{N}$ there exists a non-decreasing  function 
$K_{\ell}:  \left[ 0 , \infty \right[ \rightarrow \left[ 0 , \infty \right[$  satisfying
$
\left\Vert  L_{\ell}  \left( t \right) x  \right\Vert^{2} 
\leq K_{\ell} \left( t \right) \left\Vert  x  \right\Vert_{C}^{2}
$
for all $x \in \mathcal{D}\left( C\right)$ and $t \geq 0$.

\item[(H4.3)] There exists a non-decreasing  function 
$\alpha:  \left[ 0 , \infty \right[ \rightarrow \left[ 0 , \infty \right[$ 
and a core $\mathfrak{D}_{1}$ of $C^{2}$
such that for any $ x\in \mathfrak{D}_{1}$ we have
\[
2\Re\left\langle C^{2} x, G \left( t \right) x\right\rangle 
+\sum_{\ell=1}^{\infty }\left\Vert C L_{\ell} \left( t \right) x \right\Vert ^{2}
\leq \alpha \left( t \right)  \left\Vert x\right\Vert_{C}^{2}
\qquad
\forall t \geq 0 .
\]

\item[(H4.4)] There exists a non-decreasing function $\beta:  \left[ 0 , \infty \right[ \rightarrow \left[ 0 , \infty \right[$
and a core $\mathfrak{D}_{2}$ of $C$ such that 
$$
2\Re\left\langle  x, G \left( t \right) x\right\rangle 
+\sum_{\ell=1}^{\infty }\left\Vert  L_{\ell} \left( t \right) x \right\Vert ^{2}
\leq \beta \left( t \right) \left\Vert  x  \right\Vert^{2} 
\qquad
\forall t \geq 0 \text{ and } \forall x \in \mathfrak{D}_{2}.
$$

\end{itemize}
\end{condition}

\begin{theorem}
\label{th:ExistUniqSSE}
In addition to Hypothesis \ref{hyp:CF-inequality},
we assume that $\xi \in L_{C}^{2}\left( \mathbb{P},\mathfrak{h}\right) $ is
$\mathfrak{F}_{0}$-mea\-surable. 
Then (\ref{eq:SSE}) has a unique strong $C$-solution 
$\left( X_{t}\left( \xi \right) \right) _{t \geq0}$ with initial condition $\xi$. 
Moreover, 
\[ 
\mathbb{E} \left\Vert C  X_{t}\left( \xi \right) \right\Vert ^{2} 
\leq 
K \left( t  \right) \left( \mathbb{E} \left\Vert C  \xi  \right\Vert ^{2} 
+  \mathbb{E} \left\Vert  \xi  \right\Vert ^{2} \right).
\]
\end{theorem}

\begin{proof}
Our assertion can be be proved in much the same way as Theorem 2.4 of \cite{FagMora2013}.
\end{proof}

\begin{remark}
Theorem \ref{th:ExistUniqSSE} given above asserts that Theorem 2.4 of \cite{FagMora2013}
still holds if we replace the assumption H2.4 of \cite{FagMora2013} by Hypothesis H4.4.
We will apply Theorem \ref{th:UniquenessLinearHE} to the case:
$L_1 =  \sqrt{ 2 \kappa} a^\dagger$, $L_2  =  \sqrt{ \gamma \left(1-d \right) } \sigma^{+}$, 
$L_3  =  \sqrt{ \gamma \left(1 +d \right) } \sigma^{-} $
and 
$
G \left( t \right) 
=
\mathrm{i} H \left( t \right) - \frac{1}{2} \sum_{\ell = 1}^3 L_{\ell} L_{\ell}^*
$
with
$$
H \left( t \right) 
=
 \frac{\omega}{2} \left( 2 a^\dagger a  +\sigma^{3} \right) 
+ \mathrm{i} g \left( \alpha \left( t \right) a^\dagger -  \overline{\alpha \left( t \right)} a \right)
+ \mathrm{i} g \left( \overline{\beta  \left( t \right) } \sigma^{-}  - \beta  \left( t \right) \sigma^+ \right) .
$$
Since 
 $
G \left( t \right) + G \left( t \right)^{\ast} 
+ \sum_{\ell = 1}^3  L_{\ell}^{\ast} L_{\ell}  
=
 4 \kappa ^2 I + 2 \gamma ^2 \left( 1 + d^2 \right) \sigma_3 ,
$
Condition H2.4 of Theorem 2.4 of \cite{FagMora2013} does not apply to our situation.
\end{remark}


\section{Auxiliary equations}
\label{sec:AuxiliaryEquations}


\subsection{Auxiliary linear quantum master equation}

This subsection deals with the linear evolution equation 
obtained by replacing in  (\ref{eq:Laser1})
the unknown functions 
$t \mapsto g \, \Tr\left( \sigma^{-}  \rho_t  \right) $ and $t \mapsto g \, \Tr\left( a \,  \rho_t  \right)$
by general functions $\alpha, \beta : \left[ 0 , \infty \right[ \rightarrow \mathbb{C}$.
More precisely,
we study the well-posedness of the linear quantum master equation 
\begin{equation}
\label{eq:8.4}
\frac{d}{dt}  \rho_t
=
\mathcal{L}_{\star}^h \, \rho_t
+
\left[ 
\alpha \left( t \right) a^\dagger -  \overline{\alpha \left( t \right)} a 
+ \overline{\beta  \left( t \right) } \sigma^{-}  - \beta  \left( t \right) \sigma^+  ,
\rho_t \right] ,
\end{equation}
where 
$ \rho_t \in \mathfrak{L}_{1}^{+}\left( \ell^2(\mathbb{Z}_+)\otimes \mathbb{C}^2 \right) $,
 \begin{eqnarray}
 \label{eq:3.21}
 \mathcal{L}_{\star}^h \, \varrho
 & = &
 \left[ 
- \frac{\mathrm{i} \omega}{2} \left( 2 a^\dagger a  +\sigma^{3} \right)  , \varrho \right] 
+  2 \kappa \left( a\,\varrho a^\dagger - \frac{1}{2}  a^\dagger a \varrho - \frac{1}{2} \varrho a^\dagger a\right) 
 \\
 \nonumber
 & &
 +  \gamma(1-d)\left( \sigma^{-} \varrho \,\sigma^+ 
-\frac{1}{2} \sigma^+\sigma^{-}\varrho 
-\frac{1}{2}\varrho \,\sigma^+\sigma^{-}\right) 
\\
\nonumber
& &
+  \gamma(1+d)\left( \sigma^+\varrho \,\sigma^{-} 
-\frac{1}{2} \sigma^{-}\sigma^+\varrho 
-\frac{1}{2} \varrho \,\sigma^{-}\sigma^+\right) ,
 \end{eqnarray}
$d\in \left]-1,1 \right[$, $\omega \in\mathbb{R}$ and  $\kappa,\gamma>0$.
Furthermore,
we represent (\ref{eq:8.4}) by using 
\begin{equation}
\label {eq:SSEp}
X_{t}\left( \xi \right) 
= \xi +\int_{0}^{t}G \left( s \right) X_{s}\left( \xi \right) ds 
+ \sum_{\ell=1}^{3 }\int_{0}^{t}
L_\ell \left( s \right) X_{s}\left( \xi \right) dW_{s}^{\ell} ,
\end{equation}
where $X_{t}\left( \xi \right) \in \ell^2(\mathbb{Z}_+)\otimes \mathbb{C}^2$,
$W^1, W^2,  W^3$ are real valued independent Wiener processes,
$L_1  =  \sqrt{2 \kappa} \, a$, 
$L_2 =  \sqrt{ \gamma \left(1-d \right) } \, \sigma^{-}$ ,
$L_3 =  \sqrt{ \gamma \left(1 +d \right) } \, \sigma^{+}$
and 
$G \left( t \right)   = -  \mathrm{i} H \left( t \right) - \frac{1}{2} \sum_{\ell = 1}^3 L_{\ell}^* L_{\ell}$
with
$$
H \left( t \right) 
 =
 \frac{\omega}{2} \left( 2 a^\dagger a  +\sigma^{3} \right) 
+ \mathrm{i}  \left( \alpha \left( t \right) a^\dagger -  \overline{\alpha \left( t \right)} a \right)
+ \mathrm{i}  \left( \overline{\beta  \left( t \right) } \sigma^{-}  - \beta  \left( t \right) \sigma^+ \right).
$$

Though the open quantum system (\ref{eq:8.4}) deserves attention in its own right,
our main objective is to develop key tools for proving the results of Section \ref{sec:QuantumBifurcation}.
First, combining Theorems \ref{teor10}, \ref{th:UniquenessLinearHE}  and \ref{th:ExistUniqSSE}
we  obtain the existence and uniqueness of the regular solution to (\ref{eq:8.4}).

\begin{theorem}
 \label{th:EyU-Lineal}
Consider (\ref{eq:8.4}) with $\alpha, \beta : \left[ 0 , \infty \right[ \rightarrow \mathbb{C}$ continuous.
Let $\varrho$ be $N^p $-regular, where  $p \in \mathbb{N}$.
Then,
there exists a unique $N^p $-weak solution $\left( \rho_t \right)_{t \geq 0}$ to (\ref{eq:8.4})
with initial datum $\rho_0 = \varrho$.
Moreover, for any $t \geq 0$ we have
\begin{equation}
\label{3.11n}
\rho_{t} =\rho_0 + \int_{0}^{t} \left(
\mathcal{L}_{\star}^h \, \rho_s
+
\left[ 
 \alpha \left( s \right) a^\dagger -  \overline{\alpha \left( s \right)} a 
+  \overline{\beta  \left( s \right) } \sigma^{-}  - \beta  \left( s \right) \sigma^+  ,
\rho_s \right]
\right) ds 
\end{equation}
and
\begin{equation}
\label{eq:RepProb}
 \rho_t = \mathbb{E} \left|X_{t} \left( \xi \right)\right\rangle \left\langle X_{t} \left( \xi \right)\right| 
\qquad
\forall t \geq 0 , 
\end{equation}
where the integral appearing in (\ref{3.11n}) is understood in the sense of Bochner integral in 
$\mathfrak{L}_{1}\left( \ell^2(\mathbb{Z}_+)\otimes \mathbb{C}^2 \right) $,
$\xi \in L_{N^p}^{2}\left( \mathbb{P}, \ell^2(\mathbb{Z}_+)\otimes \mathbb{C}^2 \right) $ satisfies 
$
\varrho = \mathbb{E}\left\vert \xi \rangle \langle \xi \right\vert 
$
and
$X_{t} \left( \xi \right)$ is the unique strong $N^p$-solution of (\ref{eq:SSEp}).
\end{theorem}

\begin{proof}
 Deferred to Section  \ref{sec:Proofth:EyU-Lineal}.
\end{proof}

\begin{remark}
Assume the framework of Theorem \ref{th:EyU-Lineal}.
From the proof of Theorem \ref{th:EyU-Lineal} it follows that
$
\mathbb{E} \left\Vert X_{t} \left( \xi \right) \right\Vert_{N^p}^2 \leq K \left( t \right) \mathbb{E}  \left\Vert \xi  \right\Vert_{N^p}^2 
$
for all $t \geq 0$.
In the operator language we have 
$
\Tr\left( N^p \, \rho_{t} \, N^p \right) 
\leq
K \left( t \right) \left( 1 + \Tr\left( N^p \, \rho_{0} \, N^p \right) \right)
$
(see, e.g., \cite{MoraAP})
since 
$
\mathbb{E} \left\Vert X_{t} \left( \xi \right) \right\Vert^2 = \mathbb{E}  \left\Vert \xi  \right\Vert^2 =1
$
(see, e.g., \cite{FagMora2013}).
\end{remark}

Using the Ehrenfest-type theorem given in \cite{FagMora2013}
we get a system of  ordinary differential equations that describes 
the evolution of $\Tr\left( \rho_{t} \, a \right)$, $ \Tr\left( \rho_{t} \, \sigma^{-}  \right)$ and $\Tr\left( \rho_{t} \, \sigma^{3}  \right)$.

\begin{theorem}
\label{th:Ehrenfest-Lineal}
Under the assumptions and notation of Theorem \ref{th:EyU-Lineal},
\begin{eqnarray}
\label{eq:8.10}
 \frac{d}{dt} \Tr\left( \rho_{t} \, a \right)
 & =  &
 - \left( \kappa + \mathrm{i} \omega \right)  \Tr\left( \rho_{t} \, a \right) +  \alpha \left( t \right) ,
 \\
 \label{eq:8.11}
  \frac{d}{dt} \Tr\left( \rho_{t} \, \sigma^{-}  \right)
 & =  &
 - \left( \gamma + \mathrm{i} \omega \right)   \Tr\left( \rho_{t} \, \sigma^{-}  \right) 
 +  \beta  \left( t \right)  \Tr\left( \rho_{t} \, \sigma^{3}  \right) ,
 \\
 \label{eq:8.12}
 \frac{d}{dt} \Tr\left( \rho_{t} \, \sigma^{3}  \right)
  & =  &
 - 2  \left(
 \overline{\beta  \left( t \right)}   \Tr\left( \rho_{t} \sigma^{-}  \right) 
 + \beta  \left( t \right) \overline{  \Tr\left( \rho_{t} \sigma^{-}  \right) } 
 \right)
 \\
 \nonumber
 && \quad
 - 2 \gamma \left(  \Tr\left( \rho_{t} \sigma^{3}  \right) -d \right) .
 \end{eqnarray}
\end{theorem}

\begin{proof}
 Deferred to Section  \ref{sec:Proofth:Ehrenfest-Lineal}.
\end{proof}

\subsection{Complex Lorenz equations}
\label{sec:LorenzEquations}

Taking 
$ A \left( t \right) = \Tr\left( a \, \rho_{t} \right) $,
$ S \left( t  \right)  = \Tr\left( \sigma^{-} \rho_{t}  \right) $
and
$ D \left( t  \right) = \Tr\left(  \sigma^{3}  \rho_{t} \right) $
we rewrite (\ref{eq:laser2}) as
\begin{equation}
 \label{eq:Lorenz}
 \left\{ 
 \begin{array}{lcl} 
  \frac{d}{dt} A \left( t \right)
 & = &
 - \left( \kappa + \mathrm{i} \omega \right) A \left( t \right)  + g \ S \left( t \right) 
 \\
  \frac{d}{dt} S \left( t  \right)
&  =  &
 - \left( \gamma + \mathrm{i} \omega \right)   S \left( t  \right)
 + g \ A  \left( t \right)   D \left( t  \right)
 \\
 \frac{d}{dt} D \left( t  \right)
&  =  &
-  4 g \ \Re \left(
 \overline{ A  \left( t \right)}   S \left( t  \right) 
 \right)
 - 2 \gamma \left(  D \left( t  \right) - d \right) 
\end{array}
  \right. ,
\end{equation}
where $t \geq 0$, $D \left( t \right)  \in \mathbb{R}$ and 
$ A \left( t \right), Y \left( t \right)  \in \mathbb{C}$.
The complex Lorenz equation (\ref{eq:Lorenz})
has received much attention in the physical literature (see, e.g., \cite{Fowler1982,NingHaken1990})
due to its important role in the description of laser dynamics.
Just for the sake of completeness,
we next present  relevant properties of (\ref{eq:Lorenz}),
together with their mathematical proofs. 

\begin{theorem}
 \label{th:LorenzEquations-Laser}
 Suppose that $d\in \left]-1,1 \right[$, $\omega \in\mathbb{R}$, $ g \in \mathbb{R} \smallsetminus \left\{ 0 \right\}$
 and  $\kappa,\gamma>0$.
Then,
for every initial condition  
$A \left( 0 \right) \in \mathbb{C}$, $S \left( 0 \right) \in \mathbb{C}$, $D \left( 0 \right) \in \mathbb{R}$
there exists a unique solution 
 defined on $\left[ 0 , + \infty \right[$
 to the system (\ref{eq:Lorenz}).
 Moreover, we have:
 
\begin{itemize}

\item  If $d < 0$, then for all $t \geq 0$,
\begin{eqnarray}
  \label{eq:L5}
& & 4 \left\vert d \right\vert \left\vert A \left( t  \right) \right\vert ^2 
+ 4 \left\vert S \left( t  \right)  \right\vert ^2
+ \left( D \left( t \right) - d \right)^2
\\
\nonumber
& & \leq 
\hbox{\rm e}^{
-  2 t\, \min \left\{  \kappa  ,  \gamma \right\} }
\left(
4 \left\vert d \right\vert \left\vert A \left( 0  \right) \right\vert ^2 + 4 \left\vert S \left( 0  \right)  \right\vert ^2
+  \left( D \left( 0 \right) - d \right)^2
\right) .
\end{eqnarray}

\item If $d \geq 0$, then for any $t \geq 0$,
 \begin{eqnarray} 
 \label{eq:L4}
 && 
 \left\vert A \left( t  \right) \right\vert ^2 + \frac{g^2}{\gamma \kappa} \left\vert S \left( t  \right)  \right\vert ^2
+ \frac{g^2}{4 \gamma \kappa} \left( D \left( t \right) - d \right)^2
\\
\nonumber
&& \leq
\hbox{\rm e}^{
- t \min \left\{ \kappa - \frac{g^2d}{\gamma}   ,   \gamma - \frac{g^2d}{\kappa}   \right\} }
\left(
\left\vert A \left( 0  \right) \right\vert ^2 + \frac{g^2}{\gamma \kappa} \left\vert S \left( 0  \right)  \right\vert ^2
+ \frac{g^2}{4 \gamma \kappa} \left( D \left( 0 \right) - d \right)^2
\right) .
\end{eqnarray}

\end{itemize}
\end{theorem}

\begin{proof}
Deferred to Subection  \ref{sec:ProofTh-LorenzEquations-Laser}.
\end{proof}

\begin{remark}
According to 
$\gamma = \left( \kappa_+ + \kappa_- \right)/2$,
$d = \left( \kappa_+ - \kappa_- \right) / \left( \kappa_+ + \kappa_- \right) $
we have 
$\kappa_- = \gamma \left(1-d \right)$
and $\kappa_+ = \gamma \left(1+d \right)$.
Since 
$\kappa_+, \kappa_- > 0$,
$\gamma > 0$ and $d\in \left]-1,1 \right[$.
\end{remark}

\section{Proofs}
\label{sec:Proofs}

\subsection{Proofs of theorems from  Section \ref{sec:LinearQMEs}}

\subsubsection{Proof of Theorem \ref{teor10}}
\label{sec:ProofTh-teor10}

The proof of Theorem \ref{teor10} follows from combining 
Lemma \ref{lema18}, given below, with
the arguments used in the proof of  Theorem 4.4 of \cite{MoraAP}.
First, we get  the weak continuity  of the map $t\mapsto AX_{t}\left( \xi\right) $
in case $A$ is relatively bounded by $C$. 
 
\begin{lemma}
\label{lema17}
Let Condition H2.3 of Hypothesis \ref{hyp:formal-conservativity} hold. 
Suppose that $\xi \in L_{C}^{2}\left( \mathbb{P},\mathfrak{h}\right) $ 
and $A \in \mathfrak{L}\left( \left( \mathcal{D}\left( C\right) ,\left\Vert \cdot\right\Vert_{C}\right) ,\mathfrak{h}\right) $.
Then,
for any $\psi\in L^{2}\left( \mathbb{P},\mathfrak{h}\right) $ and $t\geq0$ we have
\begin{equation}
\lim_{s\rightarrow t}\mathbb{E}\left\langle \psi,AX_{s}\left( \xi\right)
\right\rangle =\mathbb{E}\left\langle \psi,AX_{t}\left( \xi\right)
\right\rangle.  \label{n3.9}
\end{equation}
\end{lemma}

\begin{proof}
Consider a sequence of non-negative real numbers $\left( s_{n}\right) _{n}$
satisfying $s_{n} \rightarrow t$ as $n \rightarrow + \infty$.
Since
$\left( \left( X_{s_{n}}\left( \xi\right) ,AX_{s_{n}}\left( \xi\right) ,CX_{s_{n}}\left( \xi\right) \right) \right)_{n}$ 
is a bounded sequence in $L^{2}\left( \mathbb{P}, \mathfrak{h}^{3} \right) $,
where $\mathfrak{h}^{3} = \mathfrak{h}\times \mathfrak{h}\times  \mathfrak{h}$,
there exists a subsequence $\left( s_{n\left( k\right) }\right) _{k}$ such that
\begin{equation}
 \label{n3.8}
\left( X_{s_{n\left( k\right) }}\left( \xi\right) ,AX_{s_{n\left( k\right)}}\left( \xi\right) ,CX_{s_{n\left( k\right) }}\left( \xi\right) \right)
\longrightarrow_{k\rightarrow\infty}
\left( Y,U,V\right) 
\end{equation}
weakly in $L^{2}\left( \mathbb{P},\mathfrak{h}^{3}\right)$.
Define $\mathfrak{M}=\left\{ \left( \eta,A\eta ,C\eta\right) :\eta\in L_{C}^{2}\left( \mathbb{P},\mathfrak{h}\right) \right\} $.
Thus,
 $$
 \left(X_{s_{n\left( k\right) }}\left( \xi\right) ,AX_{s_{n\left( k\right) }}\left( \xi\right) ,CX_{s_{n\left( k\right) }}\left( \xi\right) \right) 
 \in \mathfrak{M}
 \hspace{1cm}
\forall k \in \mathbb{N} .
 $$ 
 Since 
$\mathfrak{M}$ is a linear manifold of $L^{2}\left( \mathbb{P},\mathfrak{h}^{3}\right) $ 
closed with respect to the strong topology (see, e.g., proof of Lemma 7.15 of \cite{MoraAP}),
(\ref{n3.8})  implies 
$\left( Y,U,V\right) \in\mathfrak{M}$ (see, e.g., Section III.1.6 of \cite{Kato}). 
Using
$
 \mathbb{E} \left( \sup_{s \in \left[0, t+1 \right]}  \left\Vert  X_s\left( \xi \right)  \right\Vert ^{2} \right)
 < \infty
$,
together with the dominated convergence theorem
we obtain that
$$
\mathbb{E} \left\Vert X_{s_{n\left( k\right)}}\left( \xi\right) -X_{t}\left( \xi\right) \right\Vert^{2}
\rightarrow 0
\qquad \text{ as } k \rightarrow + \infty .
$$
Hence 
$Y=X_{t}\left( \xi\right) $, and so
$U=AX_{t}\left( \xi\right) $.
Therefore,
$
AX_{s_{n\left( k\right) }}\left( \xi\right) 
$
converges to
$AX_{t} \left( \xi\right) $
weakly in $L^{2}\left( \mathbb{P},\mathfrak{h}\right)$.
\end{proof}

\begin{lemma}
\label{lema18}
Assume Hypothesis \ref{hyp:formal-conservativity},
together with  
$\xi \in L_{C}^{2}\left( \mathbb{P},\mathfrak{h}\right) $ and $A \in \mathfrak{L}\left( \mathfrak{h}\right) $.
Then, 
$ t  \mapsto L_{k}  \left( t \right) X_{t}\left( \xi\right) $ is continuous as a map from 
$\left[ 0 , + \infty \right[$ to $L^{2}\left( \mathbb{P},\mathfrak{h}\right)$.
Moreover, 
\begin{eqnarray*}
t 
& \mapsto &
\mathbb{E}\left\langle G \left( t \right) X_{t}\left( \xi\right) ,A X_{t}\left( \xi\right) \right\rangle
+
\mathbb{E}\left\langle X_{t}\left( \xi\right) ,A G  \left( t \right) X_{t}\left( \xi\right) \right\rangle
\\
& & 
+\sum_{\ell=1}^{\infty }
\mathbb{E}\left\langle L_{\ell}  \left( t \right) X_{t}\left( \xi\right) ,A L_{\ell}  \left( t \right) X_{t}\left( \xi\right) \right\rangle 
\end{eqnarray*}
is a continuous function.
\end{lemma}

\begin{proof}
Suppose that $\left( t_{n}\right) _{n}$ is a sequence of non-negative real numbers satisfying
$t_n \rightarrow t$ as $n \rightarrow + \infty$.
By 
$
 \mathbb{E} \left( \sup_{s \in \left[0, t+1 \right]}  \left\Vert  X_s\left( \xi \right)  \right\Vert ^{2} \right)
 < \infty
$
(see, e.g., Th. 4.2.5 of  \cite{Prevot}),
using the dominated convergence theorem gives 
$$
 \mathbb{E} \left\Vert X_{t_{n}}\left( \xi\right) - X_{t}\left( \xi\right) \right\Vert^2 
 \longrightarrow_{n \rightarrow + \infty} 0 ,
$$
and hence 
$
AX_{t_{n}}\left( \xi\right) \longrightarrow_{n\rightarrow\infty}AX_{t}\left( \xi\right)
$
in $L^{2}\left( \mathbb{P},\mathfrak{h}\right)$.
For any $\psi \in L^{2}\left( \mathbb{P},\mathfrak{h}\right) $,
\begin{eqnarray*}
&&  \left\vert 
 \mathbb{E}\left\langle \psi, G \left( s\right) X_{s}\left( \xi\right)  \right\rangle
 - 
 \mathbb{E}\left\langle \psi, G \left( t\right) X_{t}\left( \xi\right)  \right\rangle
\right\vert
\\
& \leq & 
\mathbb{E} \left\Vert \psi \right\Vert   \left\Vert G \left( s \right) X_{s}\left( \xi\right) - G \left( t\right) X_{s}\left( \xi\right) \right\Vert
\\
& &
 +
 \left\vert 
 \mathbb{E}\left\langle \psi, G \left( t\right) X_{s}\left( \xi\right)  \right\rangle
 - 
 \mathbb{E}\left\langle \psi, G \left( t\right) X_{t}\left( \xi\right)  \right\rangle
\right\vert ,
\end{eqnarray*}
and so combining Lemma \ref{lema17}  with
$$
\mathbb{E}   \left\Vert G \left( s \right) X_{s}\left( \xi\right) - G \left( t\right) X_{s}\left( \xi\right) \right\Vert ^2
\leq
f_0 \left( s, t \right) \mathbb{E}   \left\Vert X_{s}\left( \xi\right) \right\Vert_C ^2 
$$
yields
\begin{equation}
\label{eq:13.1}
\lim_{s\rightarrow t}\mathbb{E}\left\langle \psi, G \left( s \right) X_{s}\left( \xi\right)
\right\rangle =\mathbb{E}\left\langle \psi, G \left( t \right)X_{t}\left( \xi\right)
\right\rangle .
\end{equation}
Therefore
\begin{equation}
\label{13.8}
\lim_{n\rightarrow\infty} \mathbb{E}\left\langle G \left( t_{n} \right) X_{t_{n}}\left( \xi\right) ,AX_{t_{n}}\left( \xi\right) \right\rangle
 =
 \mathbb{E}\left\langle G  \left( t \right) X_{t}\left( \xi\right) ,AX_{t}\left( \xi\right) \right\rangle
\end{equation}
(see, e.g., Section III.1.7 of \cite{Kato}). Analysis similar to that in  (\ref{eq:13.1}) shows 
$$
 \lim_{s\rightarrow t}\mathbb{E}\left\langle \psi, L_{\ell} \left( s \right) X_{s}\left( \xi\right)
\right\rangle 
=
\mathbb{E}\left\langle \psi, L_{\ell} \left( t \right)X_{t}\left( \xi\right) \right\rangle ,
$$
and hence 
\begin{equation}
\label{eq:13.2}
 L_{\ell}  \left( t_n \right) X_{t_{n}}\left( \xi\right) 
\longrightarrow_{n \rightarrow \infty}
L_{\ell}  \left( t \right) X_{t}\left( \xi\right)
\qquad
\text{weakly in } \ L^{2}\left( \mathbb{P},\mathfrak{h}\right) .
\end{equation}

According to (\ref{13.8})  with $A$ replaced by $A ^{*}$ we have the continuity of the function
$t \mapsto \mathbb{E}\left\langle A ^{*}X_{t}\left( \xi\right) ,G  \left( t \right) X_{t}\left( \xi\right) \right\rangle $,
and so 
$t \mapsto \mathbb{E}\left\langle X_{t}\left( \xi\right) ,AG  \left( t \right) X_{t}\left( \xi\right) \right\rangle $
is continuous. 
Moreover,
taking $A=I$ in (\ref{13.8}) we deduce that
$$
\mathbb{E} \Re\left\langle X_{t_{n}}\left( \xi\right) ,G   \left( t_n \right) X_{t_{n}}\left( \xi\right) \right\rangle 
\rightarrow_{n\rightarrow\infty}
\mathbb{E} \Re\left\langle X_{t}\left( \xi\right) ,G   \left( t \right) X_{t}\left( \xi\right) \right\rangle  .
$$
Applying Condition H2.2  we now get
\begin{equation}
\label{3.13}
\sum_{ \ell =1}^{\infty}
\mathbb{E}\left\Vert L_{\ell}  \left( t_n \right) X_{t_{n}}\left( \xi\right) \right\Vert ^{2}
\longrightarrow_{n\rightarrow\infty}
\sum_{\ell =1}^{\infty } \mathbb{E}\left\Vert L_{\ell}  \left( t \right) X_{t}\left( \xi\right) \right\Vert ^{2}.
\end{equation}
Combining (\ref{eq:13.2}) and (\ref{3.13}) yields 
$$
 \limsup_{n \rightarrow \infty} \mathbb{E}\left\Vert L_{\ell}  \left( t_n \right) X_{t_{n}}\left( \xi\right) \right\Vert ^{2}
\leq
\mathbb{E}\left\Vert L_{\ell}  \left( t \right) X_{t}\left( \xi\right) \right\Vert ^{2} 
$$
(see, e.g., proof of Lemma 7.16 of  \cite{MoraAP} for details)
which, together with  (\ref{eq:13.2}), implies that 
$L_{\ell} \left( t_n \right) X_{t_{n}}\left( \xi\right) $ converges strongly in $L^{2}\left( \mathbb{P},\mathfrak{h}\right)$ 
to $L_{\ell}  \left( t \right) X_{t}\left( \xi\right) $ as $n \rightarrow \infty$.
Therefore, 
$ t  \mapsto L_{\ell}  \left( t \right) X_{t}\left( \xi\right) $ is continuous as a function from 
$\left[ 0 , + \infty \right[$ to $L^{2}\left( \mathbb{P},\mathfrak{h}\right)$.

Using Condition H2.2 we obtain that
$
\sum_{ \ell = 1}^{n}
\mathbb{E}\left\langle L_{\ell}  \left( t \right) X_{t}\left( \xi\right) , AL_{\ell}  \left( t \right) X_{t}\left( \xi\right) \right\rangle
$ 
converges to
$
\sum_{\ell = 1}^{\infty}
\mathbb{E}\left\langle L_{\ell}  \left( t \right) X_{t}\left( \xi\right) , AL_{\ell}  \left( t \right) X_{t}\left( \xi\right) \right\rangle
$ 
as $n\rightarrow\infty$ uniformly on any finite interval.
Since
$$
\mathbb{E}\left\langle L_{\ell}  \left( t_n \right) X_{t_{n}}\left( \xi\right) ,AL_{ \ell}  \left( t_n \right) X_{t_{n}}\left( \xi\right) \right\rangle
\longrightarrow_{n\rightarrow\infty }
\mathbb{E}\left\langle L_{\ell}  \left( t \right) X_{t}\left( \xi\right) ,AL_{\ell}  \left( t \right) X_{t}\left( \xi\right) \right\rangle ,
$$
the map 
$
t \mapsto 
\sum_{\ell=1}^{\infty}
\mathbb{E}\left\langle L_{\ell}  \left( t \right) X_{t}\left( \xi\right) ,AL_{\ell}  \left( t \right) X_{t}\left( \xi\right) \right\rangle
$
is continuous.
\end{proof}

\begin{lemma}
\label{lema61}
 Let Hypothesis \ref{hyp:formal-conservativity} hold,
 except Condition H2.4.
 For any $\xi \in L_{C}^{2}\left( \mathbb{P},\mathfrak{h}\right) $,
we define
\begin{eqnarray*}
 \mathcal{L}_{*} \left( \xi ,t \right)
& = &
\mathbb{E} \left\vert   G  \left( t \right)  X_{t}\left( \xi\right) \rangle  \langle X_{t}\left( \xi\right)  \right\vert 
+
 \mathbb{E}  \left\vert X_{t}\left( \xi\right)  \rangle  \langle G  \left( t \right)  X_{t}\left( \xi\right) \right\vert 
 \\
 & &
 + \sum_{\ell=1}^{\infty}\mathbb{E} \left\vert  L_{\ell}  \left( t \right)  X_{t}\left( \xi\right) \rangle  \langle L_{\ell}  \left( t \right)  X_{t}\left( \xi \right) \right\vert .
\end{eqnarray*}
Then $\mathcal{L}_{*} \left( \xi ,t \right)$ is a trace-class operator  on $\mathfrak{h}$ whose trace-norm is uniformly bounded with respect to $t$ on bounded time intervals;
the series involved in the definition of $\mathcal{L}_{*}$ converges in $\mathfrak{L} _{1}\left( \mathfrak{h} \right)$.
\end{lemma}

\begin{proof}
By Condition H2.2, 
using 
$
\left\Vert 
\left\vert x \rangle\langle y \right\vert  
\right\Vert_{1}
= 
\left\Vert x \right\Vert \left\Vert y \right\Vert
$ 
and Lemma 7.3 of \cite{MoraAP} we get
\begin{eqnarray*}
&&  \left\Vert  \mathbb{E}   \left\vert   G  \left( t \right) X_{t}\left( \xi\right) \rangle  \langle X_{t} \left( \xi\right)  \right\vert  \right\Vert_{1}
 + 
 \left\Vert  \mathbb{E} \left\vert X_{t}\left( \xi\right)  \rangle   \langle G \left( t \right) X_{t}\left( \xi\right) \right\vert \right\Vert_{1}
\\
& & \qquad
+ 
 \sum_{ \ell=1}^{\infty}
   \left\Vert \mathbb{E}  \left\vert  L_{\ell} \left( t \right) X_{t}\left( \xi\right) \rangle \langle L_{\ell} \left( t \right) X_{t}\left( \xi \right) \right\vert
  \right\Vert_{1} 
 \\
&  \leq &
4 \mathbb{E} \left( \left\Vert X_{t}\left( \xi \right) \right\Vert \left\Vert G \left( t \right) X_{t}\left( \xi \right) \right\Vert \right)
\leq
K \left( t \right) \sqrt{ \mathbb{E} \left\Vert  \xi  \right\Vert^{2} } 
\sqrt{ \mathbb{E} \left\Vert X_{t}\left( \xi \right)    \right\Vert_{C}^{2} },
\end{eqnarray*}
where the last inequality follows from Condition H2.1.
\end{proof}

Applying Lemma 7.3 of \cite{MoraAP}  and Lemma  \ref{lema18} we easily obtain Lemma \ref{lema61b}.

\begin{lemma}
 \label{lema61b}
 Suppose that Hypothesis \ref{hyp:formal-conservativity} hold, $\xi \in L_{C}^{2}\left( \mathbb{P},\mathfrak{h}\right) $,
 and $A \in \mathfrak{L} \left( \mathfrak{h} \right)$.
 Then,
 $t \mapsto   \Tr \left( A   \mathcal{L}_{*} \left( \xi ,t \right) \right)$ is continuous 
 as a function from $\left[0, \infty \right[$ to $ \mathbb{C}$,
 and 
\begin{eqnarray*}
 \Tr \left(A   \mathcal{L}_{*} \left( \xi ,t \right) \right)
 & = &
 \mathbb{E}\left\langle X_{t}\left( \xi\right) ,AG  \left( t \right)  X_{t}\left( \xi\right) \right\rangle
+
\mathbb{E}\left\langle G \left( t \right) X_{t}\left( \xi\right) ,AX_{t}\left( \xi\right) \right\rangle
\\
& &
+ \sum_{\ell =1}^{\infty }
\mathbb{E}\left\langle L_{\ell} \left( t \right) X_{t}\left( \xi\right) ,A L_{\ell} \left( t \right) X_{t}\left( \xi\right) \right\rangle .
\end{eqnarray*}
Here, 
$  \mathcal{L}_{*} \left( \xi ,t \right) $ is as in Lemma \ref{lema61}.
\end{lemma}

\begin{lemma}
\label{lema23}
 Adopt  Hypothesis \ref{hyp:formal-conservativity},
 together with  $\xi \in L_{C}^{2}\left( \mathbb{P},\mathfrak{h}\right) $.
Then 
\begin{equation}
\label{13.6}
\rho_{t}
=
 \mathbb{E}\left\vert \xi \rangle \langle \xi \right\vert 
+
\int_{0}^{t} \mathcal{L}_{*} \left( \xi ,s \right)  ds,
\end{equation}
where 
$t\geq0$ 
and
$\mathcal{L}_{*} \left( \xi ,s \right)$ is as in Lemma \ref{lema61};
we understand the above integral in the sense of Bochner integral in $\mathfrak{L}_{1}\left( \mathfrak{h}\right) $.
\end{lemma}

\begin{proof}
Fix $x \in \mathfrak{h}$,
and 
choose 
  $
 \tau_{n}=\inf\left\{ s\geq0:\left\Vert X_{s}\left( \xi\right) \right\Vert >n\right\} 
 $,
 with $n \in \mathbb{N}$.
 Applying the complex It\^{o} formula we obtain that  
\begin{equation}
\label{13.5}
 \left\langle X_{t\wedge \tau_{n}}\left( \xi\right) ,x\right\rangle X_{t\wedge \tau_{n}}\left( \xi\right)
=
\left\langle \xi,x\right\rangle \xi+\mathbb{E}\int_{0}^{t\wedge \tau_{n}} L_{x} \left(X_{s}\left( \xi\right) , s \right) ds
+
M_t,
\end{equation}
where
$$
M_t
=
\sum_{\ell = 1}^{\infty}\int_{0}^{t\wedge \tau_{n}}\left( \left\langle
X_{s}\left( \xi\right) ,x\right\rangle L_{\ell} \left( s \right) X_{s}\left( \xi\right)
+\left\langle L_{\ell} \left( s\right) X_{s}\left( \xi\right) ,x\right\rangle X_{s}\left(
\xi\right) \right) dW_{s}^{\ell}
$$
and  for any $z \in \mathcal{D}\left( C \right)$,
$$
L_{x} \left(z , s \right) 
= 
\left\langle z ,x\right\rangle G \left( s \right) z
+ \left\langle G \left( s\right) z ,x\right\rangle z
+\sum_{k=1}^{\infty}\left\langle L_{k} \left( s\right) z ,x\right\rangle L_{k} \left( s\right) z .
$$
According to Condition H2.2 we have
 \begin{eqnarray*}
&&
\mathbb{E}\sum_{\ell=1}^{\infty}\int_{0}^{t\wedge \tau_{n}}\left\Vert \left\langle X_{s}\left( \xi\right) ,x\right\rangle L_{\ell}  \left( s \right) X_{s}\left(
\xi\right) +\left\langle L_{\ell}  \left( s \right) X_{s}\left( \xi\right) ,x\right\rangle X_{s}\left( \xi\right) \right\Vert ^{2}ds 
\\
& & \leq 
4 n^{3} \left\Vert x\right\Vert ^{2}
\mathbb{E} \int_{0}^{t\wedge \tau_{n}} \left\Vert G  \left( s \right) X_{s} \right\Vert  ds.
\end{eqnarray*}
Therefore 
$
\mathbb{E} M_t =0
$
by Condition H2.1,
and 
so (\ref{13.5}) yields
\begin{equation}
\label{13.1}
\mathbb{E} \left\langle X_{t\wedge \tau_{n}}\left( \xi\right) ,x\right\rangle X_{t\wedge \tau_{n}}\left( \xi\right)
=
 \mathbb{E}\left\langle \xi,x\right\rangle \xi+\mathbb{E}\int_{0}^{t\wedge \tau_{n}} L_{x} \left(X_{s}\left( \xi\right) ,s \right) ds.
\end{equation}

We will take the limit as $n \rightarrow \infty$ in (\ref{13.1}).
Since  $ X \left( \xi\right)$ has continuous sample paths,
$\tau_{n}\nearrow_{n\rightarrow\infty}\infty$. 
By  H2.1 and H2.2,  
applying the dominated convergence yields
$$
\lim_{n \rightarrow \infty} \mathbb{E}\int_{0}^{t\wedge \tau_{n}} L_{x} \left(X_{s}\left( \xi\right) , s\right) ds
=
\mathbb{E}\int_{0}^{t} L_{x} \left(X_{s}\left( \xi\right) ,s \right) ds.
$$
Combining
$
 \mathbb{E} \left( \sup_{s \in \left[0, t+1 \right]}  \left\Vert  X_s\left( \xi \right)  \right\Vert ^{2} \right)
 < \infty
$
with the dominated convergence theorem gives 
$
  \lim_{n \rightarrow \infty}  \mathbb{E} \left\langle X_{t\wedge \tau_{n}}\left( \xi\right) , x\right\rangle X_{t\wedge \tau_{n}}\left( \xi\right)
  =
  \mathbb{E} \left\langle X_{t}\left( \xi\right) ,x\right\rangle X_{t}\left( \xi\right).
$
Then, 
letting first $n \rightarrow \infty$ in (\ref{13.1}) and then using  Fubini's theorem we get
\begin{equation}
\label{6.3}
\mathbb{E}\left\langle X_{t}\left( \xi\right) ,x\right\rangle X_{t}\left( \xi\right) 
=
\mathbb{E}\left\langle \xi,x\right\rangle \xi
+ 
\int_{0}^{t} \mathbb{E} L_{x} \left(X_{s}\left( \xi\right) , s\right).
\end{equation}

By Condition H2.2, 
the dominated convergence theorem leads to
$$
\mathbb{E} \sum_{k=1}^{\infty} \left\langle L_{k} \left( s \right) X_{s}\left( \xi\right) ,x\right\rangle L_{k}  \left( s \right) X_{s}\left( \xi\right)
=
\sum_{k=1}^{\infty}\mathbb{E} \left\langle L_{k}  \left( s \right)  X_{s}\left( \xi\right) ,x\right\rangle L_{k}  \left( s \right)  X_{s}\left( \xi\right) ,
$$
and so Lemma 7.3 of \cite{MoraAP}  yields 
$
\mathbb{E} L_{x} \left(X_{s}\left( \xi\right) ,s \right)
=
\mathcal{L}_{*} \left( \xi ,s \right) x
$,
hence
\begin{equation}
\label{6.1}
 \int_{0}^{t} \mathbb{E} L_{x} \left(X_{s}\left( \xi\right) , s \right)
=
\int_{0}^{t} \mathcal{L}_{*} \left( \xi ,s \right) x  ds .
\end{equation}

Since  the dual of $ \mathfrak{L}_{1}\left( \mathfrak{h}\right)$ consists in all linear maps 
$\varrho \mapsto \Tr \left( A \varrho \right)$ with $A \in  \mathfrak{L}\left( \mathfrak{h}\right)$,
Lemma \ref{lema61b} implies that
$
t \mapsto \mathcal{L}_{*} \left( \xi ,t \right)
$
is measurable as a function from $\left[0, \infty \right[$
to $ \mathfrak{L}_{1}\left( \mathfrak{h}\right)$.
Furthermore,
using Lemma \ref{lema61} we get that 
$
t \mapsto \mathcal{L}_{*} \left( \xi ,t \right)
$
is a Bochner integrable $\mathfrak{L}_{1}\left( \mathfrak{h}\right) $-valued function on bounded intervals.
Then
(\ref{6.3}), together with  (\ref{6.1}), 
gives (\ref{13.6}).
\end{proof}

\begin{proof}{\it (of Theorem \ref{teor10})}
According to Theorem 3.2  of \cite{MoraAP} we have
$$
A G \left( t \right) \rho_{t}  
= 
\mathbb{E} \left\vert  A G  \left( t \right) X_{t} \left( \xi\right) \rangle\langle  X_{t}\left( \xi\right) \right\vert .
$$
Since $G \left( t \right), L_{1}  \left( t \right), L_{2}  \left( t \right), \ldots $ are closable, 
$G \left( t \right)^{\ast}, L_{1} \left( t \right)^{\ast}, L_{2} \left( t \right)^{\ast}, \ldots $ are densely defined 
and $G \left( t \right)^{\ast \ast}$, $L_{1} \left( t \right)^{\ast \ast}, \ldots$  coincide with the closures of  $G \left( t \right), L_{1} \left( t \right), \ldots$ respectively 
(see, e.g., Theorem III.5.29 of \cite{Kato}).
Now,
Theorem 3.2  of \cite{MoraAP}   yields
$
A \rho_{t} G \left( t \right)^{\ast}  = \mathbb{E} \left\vert  AX_{t}\left( \xi\right) \rangle\langle G \left( t \right)X_{t}\left( \xi\right)  \right\vert
$
and
$$
 A L_{k} \left( t \right) \rho_{t} L_{k} \left( t \right)^{\ast} 
 =
\mathbb{E} \left\vert A L_{k} \left( t \right)X_{t}\left( \xi\right) \rangle\langle  L_{k} \left( t \right) X_{t}\left( \xi\right) \right\vert .
$$
Therefore
\begin{equation}
\label{eq:6.4}
 \mathcal{L}_{*} \left( \xi ,t \right)
=
G  \left( t \right) \rho_{t}  +  \rho_{t} G \left( t \right)^{\ast} 
 +\sum_{k=1}^{\infty} L_{k} \left( t \right) \rho_{t} L_{k} \left( t \right)^{\ast} ,
\end{equation}
where $\mathcal{L}_{*} \left( \xi ,t \right)$  is as in Lemma \ref{lema61}.
Combining (\ref{eq:6.4}) with Lemma \ref{lema23} we get (\ref{3.11}),
and so 
$
\Tr\left( A\rho_{t} \right) 
=
\Tr\left( A \varrho \right) 
+
\int_{0}^{t} 
\Tr\left(
A \mathcal{L}_{*} \left( \xi ,s \right)
\right) ds
$
for all $t \geq 0$.
Using the continuity of $\mathcal{L}_{*} \left( \xi ,\cdot \right)$ we obtain (\ref{3.12}).
\end{proof}

\subsubsection{Proof of Theorem \ref{th:SesqForm}}
\label{sec:ProofTh-SesqForm}

\begin{proof}
For any $x,y \in \mathcal{D}\left( C \right)$ we set
$
\lbrack x, y \rbrack
=
\mathbb{E} \left\langle X_{t} \left( x \right), A  X_{t} \left( y \right) \right\rangle 
$.
According to Definition \ref{def:regular-sol} we have
 $$
 \left\vert \lbrack x, y \rbrack \right\vert
 =
  \left\vert  \mathbb{E} \left\langle X_{t} \left( x \right), A  X_{t} \left( y \right) \right\rangle \right\vert
  \leq
  K \left( t \right) \left\Vert A \right\Vert  \left\Vert x \right\Vert  \left\Vert y \right\Vert 
  \qquad
 \forall x,y \in \mathcal{D}\left( C \right) .
 $$
 Since $\mathcal{D}\left( C \right)$ is dense in $\mathfrak{h}$,
  $\lbrack \cdot, \cdot \rbrack$ 
 can be extended uniquely to a sesquilinear form 
 $\lbrack \cdot, \cdot \rbrack$ over $\mathfrak{h} \times \mathfrak{h}$ 
satisfying
$
 \left\vert \lbrack x, y \rbrack \right\vert
\leq
  K \left( t \right) \left\Vert A \right\Vert \left\Vert x \right\Vert  \left\Vert y \right\Vert
$
for any $x,y \in \mathfrak{h}$.
Hence,
there exists a unique bounded operator $\mathcal{T}_{t}\left( A \right)$ on $\mathfrak{h}$ such that $ \left\vert \lbrack x, y \rbrack \right\vert = \left\langle x, \mathcal{T}_{t}\left( A \right) y \right\rangle$ for all $x,y$ in $\mathfrak{h}$.
Moreover, 
$\left\Vert \mathcal{T}_{t}\left( A \right) \right\Vert \leq  K \left( t \right) \left\Vert  A \right\Vert$.
\end{proof}

\subsubsection{Proof of Theorem \ref{th:UniquenessLinearHE}}
\label{sec:ProofTh-UniquenessLinearHE}

\begin{proof}
Using  It\^o's formula
we will prove that 
for all  $x,y \in \mathcal{D}\left(  C \right)$,
\begin{equation}
 \label{4.5}
 \mathbb{E} \left\langle  X_{t}\left( x \right),  A   X_{t} \left( y \right) \right\rangle
 =
  \left\langle x, \mathcal{A}_{t} y  \right\rangle.
\end{equation}
This, together with Theorem \ref{th:SesqForm}, implies 
 $\mathcal{A}_{t} = \mathcal{T}_{t}\left( A \right)$.
 
Motivated by  $\mathcal{A}_t$ is only a weak solution,
we fix an orthonormal basis $\left( e_{n} \right)_{n \in \mathbb{N}}$ of $\mathfrak{h}$
and  consider the function 
$
F_{n} :  \left[0 , t \right] \times \mathfrak{h}  \times \mathfrak{h}  
\rightarrow 
\mathbb{C}
$ 
defined by
\[
F_{n} \left( s, u, v \right) = \left\langle  R_{n} \overline{u}, \mathcal{A}_{t-s}  R_{n} v \right\rangle,
\]
where 
$R_{n} = n \left(n+C \right)^{-1}$
and 
 $
\bar{u}
=
\sum_{n \in \mathbb{N}} \overline{\left\langle e_{n}, u \right\rangle} e_{n}
$.
Since the range of $R_{n}$ is contained in $\mathcal{D}\left( C \right)$, 
\begin{equation}
\label{4.6}
\frac{d}{ds} F_{n} \left( s, u, v \right) = - g \left( s, R_{n} \overline{u}, R_{n} v \right) ,
\end{equation}
with
$
g \left( s, x, y \right)
=
\left\langle  x, \mathcal{A}_{t-s}  G y \right\rangle
+ \left\langle   G x, \mathcal{A}_{t-s} y \right\rangle
+ \sum_{k=1}^{\infty} \left\langle   L_{k}  x, \mathcal{A}_{t-s}  L_{k}  y \right\rangle 
$.
We have that
$t \longmapsto  \left\langle u, \mathcal{A}_{t} v \right\rangle$ is continuous
for all $u,v \in \mathfrak{h}$,
and so 
combining $C R_{n} \in \mathfrak{L}\left( \mathfrak{h}\right)$ with Hypothesis \ref{hyp:Well-posed} 
we get  the uniformly continuity of 
$
\left(s, u, v \right)
 \longmapsto
g \left( s, R_{n} \overline{u}, R_{n} v \right)
$
on bounded subsets of $\left[ 0, t \right] \times \mathfrak{h} \times \mathfrak{h}$.
Then,
we can apply It\^o's formula to 
$ F_{n} \left( s \wedge\tau_{j}, \overline{X_{s}^{\tau_{j}}\left( x \right)}, X_{s}^{\tau_{j}}\left( y \right) \right) $,
with
$$
\tau _{j} = 
\inf{ \left\{  t \geq0: \left\Vert  X_{t}\left( x \right)  \right\Vert + \left\Vert  X_{t}\left( y \right)  \right\Vert> j \right\} } .
$$

Fix $x,y \in \mathcal{D}\left(  C \right)$.
Combining  It\^o's formula with (\ref{4.6})
we deduce that
\[
 F_{n} \left( t \wedge\tau_{j}, \overline{X_{t}^{\tau_{j}}\left( x \right)}, X_{t}^{\tau_{j}}\left( y \right) \right) 
  = 
  F_{n} \left( 0, \overline{X_{0}\left( x \right)}, X_{0}\left( y \right) \right)
  + I_{t \wedge\tau_{j}}^n + M_t ,
\]
where for $s \in \left[ 0, t \right]$:
\begin{eqnarray*}
 M_{s}
 & =  &
 \sum_{k = 1}^{\infty} \int_{0}^{s \wedge \tau_{j}} \left\langle 
R_{n}  X_{r} ^{\tau_{j}} \left( x \right),
 \mathcal{A}_{t-r}
 R_{n} L_{k}  X_{r}^{\tau_{j}} \left( y \right) 
\right\rangle dW^{k}_{r}
\\
& &
 +
\sum_{k = 1}^{\infty}  \int_{0}^{s \wedge \tau_{j}} \left\langle 
R_{n} L_k X_{r} ^{\tau_{j}} \left( x \right),
 \mathcal{A}_{t-r}
 R_{n}  X_{r}^{\tau_{j}} \left( y \right) 
\right\rangle dW^{k}_{r}
\end{eqnarray*}
 and 
$
I_s^n
=
\int_{0}^{s} 
  \left( 
  -  g \left( r, R_n  X_{r} \left( x \right) , R_n X_{r} \left( y \right) \right) 
   +
    g_{n}  \left( r, X_{r} \left( x \right) , X_{r} \left( y \right) \right) 
   \right)  dr
$
with
$$
g_{n} \left( r, u, v \right) 
=
\left\langle R_{n} u, \mathcal{A}_{t-r}    R_{n} G v \right\rangle
+ \left\langle R_{n}  G u, \mathcal{A}_{t-r}  R_{n} v \right\rangle
+ \sum_{k=1}^{\infty} \left\langle  R_{n} L_{k}  u, \mathcal{A}_{t-r}  R_{n} L_{k}  v \right\rangle .
$$

We next establish the martingale property of $M_s$.
For all $r \in \left[0, t \right]$ we have 
\[
\left\Vert R_{n}  X_{r} ^{\tau_{j}} \left( x \right)  \right\Vert^2
 \left\Vert \mathcal{A}_{t-r} \right\Vert^2 
  \left\Vert R_{n} L_{k}  X_{r}^{\tau_{j}} \left( y \right)  \right\Vert^2 
 \leq 
 j^2 \sup_{s \in \left[0, t \right]} \left\Vert \mathcal{A}_{s} \right\Vert ^2
 \left\Vert L_{k}  X_{r}^{\tau_{j}} \left( y \right)  \right\Vert^2 .
\]
By H2.1 and H3.1,
$
\mathbb{E}
\int_{0}^{t \wedge \tau_{j}} 
\sum_{k = 1}^{\infty}
\left|
\left\langle 
R_{n}  X_{r} ^{\tau_{j}} \left( x \right),
 \mathcal{A}_{t-r}
 R_{n} L_{k}  X_{r}^{\tau_{j}} \left( y \right) 
\right\rangle
\right|^{2} ds 
<
\infty 
$.
Thus
$$
\left(
\sum_{k = 1}^{\infty}
\int_{0}^{s \wedge \tau_{j}} \left\langle 
R_{n}  X_{r} ^{\tau_{j}} \left( x \right),
 \mathcal{A}_{t-r}
 R_{n} L_{k}  X_{r}^{\tau_{j}} \left( y \right) 
\right\rangle dW^{k}_{r}
\right)_{s \in \left[ 0, t \right]}
$$
is a martingale.
The same conclusion can be draw for
\[
\sum_{k = 1}^{\infty}
\int_{0}^{s \wedge \tau_{j}} \left\langle 
R_{n} L_k X_{r} ^{\tau_{j}} \left( x \right),
 \mathcal{A}_{t-r}
 R_{n}  X_{r}^{\tau_{j}} \left( y \right) 
\right\rangle dW^{k}_{r},
\]
and so $\left( M_s \right)_{s \in \left[ 0, t \right]}$ is a martingale. 
Hence
\begin{equation}
\label{4.15}
\mathbb{E}
  \left\langle  R_{n} X_{ t }^{\tau_{j}}\left( x \right) , 
  \mathcal{A}_{t - t \wedge\tau_{j}}  
  R_{n} X_{t}^{\tau_{j}}\left( y \right) 
 \right\rangle
 =
 \left\langle R_{n} x , \mathcal{ A}_{t} R_{n}  y  \right\rangle
 +
 \mathbb{E}  I_{t \wedge\tau_{j}}^n .
 \end{equation}

 We will take the limit as  $j \rightarrow \infty$ in  (\ref{4.15}).
 Since 
$
 \mathbb{E} \left( \sup_{s \in \left[0, t \right]}  \left\Vert  X_s\left( \xi \right)  \right\Vert ^{2} \right)
 < \infty
$
for $\xi = x, y$
(see, e.g., Th. 4.2.5 of  \cite{Prevot}),
using the dominated convergence theorem,
together with the continuity of 
$t \longmapsto  \left\langle u, \mathcal{A}_{t} v \right\rangle$,
we get
\[
\mathbb{E} \left\langle R_{n}  X_{t}^{\tau_{j}}\left( x \right),  \mathcal{A}_{t- t \wedge \tau_{j}} R_{n}  X_{t}^{\tau_{j}}\left( y \right) \right\rangle
\rightarrow_{j \rightarrow \infty} 
\mathbb{E} \left\langle R_{n}  X_{t}\left( x \right),  A  R_{n}  X_{t} \left( y \right) \right\rangle.
\]
Applying again the dominated convergence theorem yields
$
 \mathbb{E}  I_{t \wedge\tau_{j}}^n 
 \longrightarrow_{j \rightarrow \infty} 
 \mathbb{E}  I_{t }^n 
$,
and hence letting $j \rightarrow \infty$ in  (\ref{4.15}) we deduce that
 \begin{eqnarray}
\label{4.14}
&&  \mathbb{E} \left\langle R_{n} X_{t} \left( x \right), A R_{n} X_{t} \left( y \right) \right\rangle 
 -  \left\langle R_{n} x , \mathcal{ A}_{t} R_{n}  y  \right\rangle  
 \\
 & & =
 \nonumber
 \mathbb{E} \int_{0}^{t} 
  \left(
  -   g \left( s, R_n X_{s} \left( x \right) , R_n X_{s} \left( y \right) \right) 
  +  g_{n}  \left( s, X_{s} \left( x \right), X_{s} \left( y \right) \right) 
   \right)  ds.
\end{eqnarray}

Finally,
we take the limit as $n \rightarrow \infty$ in (\ref{4.14}).
Since  
$\left\Vert R_{n} \right\Vert \leq 1$
and
$ R_{n}$ tends pointwise to $I$ as $n \rightarrow \infty$,
the  dominated convergence theorem yields 
\[
\lim_{n \rightarrow \infty} \mathbb{E} \int_{0}^{t} 
  g_{n}  \left( s, X_{s} \left( x \right), X_{s} \left( y \right) \right)   ds
  =
  \mathbb{E} \int_{0}^{t} 
  g  \left( s, X_{s} \left( x \right), X_{s} \left( y \right) \right)   ds .
\]
For any $x \in \mathcal{D}\left( C \right)$, 
$ \lim_{n \rightarrow \infty} C R_{n} x = C x$.
By
$\left\Vert C R_{n} x \right\Vert \leq \left\Vert C x \right\Vert $, 
using  the dominated convergence theorem gives
\[
\lim_{n \rightarrow \infty} \mathbb{E} \int_{0}^{t} 
 g \left( s, R_n X_{s} \left( x \right) , R_n X_{s} \left( y \right) \right)
  ds
  =
  \mathbb{E} \int_{0}^{t} 
  g  \left( s, X_{s} \left( x \right), X_{s} \left( y \right) \right)   ds.
\]
Thus,
letting $n \rightarrow \infty$ in (\ref{4.14}) we obtain (\ref{4.5}).
 \end{proof}

\subsection{Proofs of theorems from  Section \ref{sec:AuxiliaryEquations}}

\subsubsection{Proof of Theorem \ref{th:EyU-Lineal}}
\label{sec:Proofth:EyU-Lineal}

 \begin{proof}
 First, we show that $ \rho_t$ given by (\ref{eq:RepProb}) is a $N^p $-weak solution to (\ref{eq:8.4}).
 To this end, 
 we will verify that $C = N^p$ satisfies Hypothesis \ref{hyp:formal-conservativity},
 where, here and subsequently,
 $H \left( t \right)$, $G \left( t \right)$, $L_1$, $L_2$, $L_3$ are defined as in Theorem \ref{th:EyU-Lineal}.
 Since  $L_2, L_3 \in \mathfrak{L}\left(  \ell^2(\mathbb{Z}_+)\otimes \mathbb{C}^2\right)$, 
 $L_1, L_1^{*}L_1$ are relatively bounded  with respect to $N$
 and 
 $$
 \left\Vert  H \left( t \right)  x  \right\Vert^2
 \leq 
 K \max \left(  \left\vert  \alpha \left( t \right) \right\vert  ,  \left\vert  \beta \left( t \right) \right\vert  \right)  
 \left\Vert  x  \right\Vert_N 
 \qquad
 \forall x \in \mathcal{D}\left( N \right) ,
 $$
 $C$ fulfills Condition H2.1 of Hypothesis \ref{hyp:formal-conservativity}. 
 By definition of $ G \left( t \right)$ and $ L_\ell $,
  $$
 2\Re\left\langle x, G \left( t \right) x\right\rangle 
+ \sum_{\ell=1}^{3}\left\Vert L_\ell \, x\right\Vert ^{2} = 0
\qquad
\forall x \in  \mathcal{D}\left( N \right) ,
$$
and hence Condition H2.2 holds.
Condition H2.4 follows from the continuity of $\alpha$ and $\beta$.

 In order to check Condition H2.3, 
 we denote by $\mathfrak{D}$ the set of all  $x \in  \ell^2(\mathbb{Z}_+)\otimes \mathbb{C}^2$ 
 such that 
 $x \left( n,\eta \right) := \left\langle e_n \otimes e_{\eta}, x \right\rangle$ 
 is equal to $0$  for all combinations of  $n \in \mathbb{Z}_{+}$ and $\eta = \pm$ except a finite number. 
 Consider $x \in \mathfrak{D}$.
 A careful computation yields
\begin{eqnarray}
  \label{eq:8.13}
 & & 2\Re\left\langle N^{2p} x, G \left( t \right) x\right\rangle +\sum_{\ell=1}^{3 }   \left\Vert N^{p} L_{\ell}x \right\Vert ^{2}
 \\
 \nonumber
 & & = 
 \sum_{k \in \mathbb{Z}_{+}, \eta = \pm} 
  2  \, \Re \left( \alpha \left( t \right)  x \left( k, \eta \right) \overline{ x \left(k+1, \eta \right) }  \right)
  \sqrt{k+1} \left(  \left( k+1 \right)^{2p} - k^{2p} \right)
  \\
  \nonumber
  && \quad \quad \quad  +
  \sum_{k \in \mathbb{Z}_{+}, \eta = \pm}  2 \kappa \left\vert  x \left( k, \eta \right) \right\vert ^2 
 k \left(  \left( k-1 \right)^{2p} - k^{2p} \right) .
 \end{eqnarray}
Since  
\begin{eqnarray*}
&& \sum_{k \in \mathbb{Z}_{+}, \eta = \pm} 
  2  \, \Re \left( \alpha \left( t \right)  x \left( k, \eta \right) \overline{ x \left(k+1, \eta \right) }  \right)
  \sqrt{k+1} \left(  \left( k+1 \right)^{2p} - k^{2p} \right)
\\
&& \leq 
  2   \left\vert    \alpha \left( t \right) \right\vert  
  \sum_{k \in \mathbb{Z}_{+}, \eta = \pm} 
    \left\vert x \left( k, \eta \right)  \right\vert  \left\vert x \left(k+1, \eta \right) \right\vert \phi \left( k \right)   
\\
& &  \leq 
 2   \left\vert    \alpha \left( t \right) \right\vert  
  \sum_{k \in \mathbb{Z}_{+}, \eta = \pm} 
    \left\vert x \left( k, \eta \right)  \right\vert^2 \phi \left( k \right) 
\end{eqnarray*}
with
$
\phi \left( k \right) 
= 
\sqrt{k+1} \left(  \left( k+1 \right)^{2p} - k^{2p} \right)
=
\sqrt{k+1} \sum_{j=0}^{2p-1} 
\left( \begin{array}{c} 2p \\ j  \end{array} \right)  
k^j
$, 
\begin{eqnarray}
 \label{eq:8.14}
&& \sum_{k \in \mathbb{Z}_{+}, \eta = \pm} 
  2  \, \Re \left( \alpha \left( t \right)  x \left( k, \eta \right) \overline{ x \left(k+1, \eta \right) }  \right)
  \sqrt{k+1} \left(  \left( k+1 \right)^{2p} - k^{2p} \right) 
\\
\nonumber
&  \leq & 
   \left\vert   \alpha \left( t \right) \right\vert  
  K \sum_{k \in \mathbb{Z}_{+}, \eta = \pm} 
    \left\vert x \left( k, \eta \right)  \right\vert^2  \left( 1 + k^{2p-1/2}  \right) .
 \end{eqnarray}
Combining (\ref{eq:8.13}) with (\ref{eq:8.14})  we get  
$$
 2\Re\left\langle N^{2p} x, G \left( t \right) x\right\rangle +\sum_{\ell=1}^{3 }   \left\Vert N^{p} L_{\ell}x \right\Vert ^{2} 
 \leq  K  \left\vert  \alpha \left( t \right) \right\vert \left\Vert x \right\Vert_{ N^{p}} ^{2},
$$
and so Condition H4.3 of Hypothesis \ref{hyp:CF-inequality} holds
because $\mathfrak{D}$ is a core of $N^p$.
Then,
applying Theorem 2.4 of \cite{FagMora2013} (see also Theorem \ref{th:ExistUniqSSE}) 
we obtain that for any  initial condition $\xi \in L_{N^p}^{2}\left( \mathbb{P}, \ell^2(\mathbb{Z}_+)\otimes \mathbb{C}^2\right) $
there exists a unique strong $N^p$-solution of (\ref{eq:SSEp}), 
together with
\begin{equation}
\label{eq:3.3}
\mathbb{E} \left\Vert X_{t} \left( \xi \right) \right\Vert_{N^p}^2 \leq K \left( t \right) \mathbb{E}  \left\Vert \xi  \right\Vert_{N^p}^2 .
\end{equation}
Therefore, Condition H2.3 holds,
and so we have checked Hypothesis \ref{hyp:formal-conservativity} with $C = N^p$.

Applying Theorem 3.1 of \cite{MoraAP}  yields 
$
\varrho = \mathbb{E}\left\vert \xi \rangle \langle \xi \right\vert 
$
for certain 
$$\xi \in L_{N^p}^{2}\left( \mathbb{P}, \ell^2(\mathbb{Z}_+)\otimes \mathbb{C}^2\right) . $$
Using Theorem \ref{teor10} we obtain that
$
 \rho_t := \mathbb{E} \left|X_{t} \left( \xi \right)\right\rangle \left\langle X_{t} \left( \xi \right)\right| 
$
satisfies the relation (\ref{3.11n}) and 
\begin{equation}
 \label{eq:8.8}
 \left\{ 
 \begin{array}{lll}
  \frac{d}{dt}\Tr\left( A\rho_{t} \right) 
  & = &
\Tr\left(
A\left( G \left( t \right) \rho_{t}\ +\rho_{t} G  \left( t \right)  ^{\ast}
+\sum_{\ell=1}^{3} L_{\ell}    \rho_{t} L_{\ell}  ^{\ast }\right) 
\right) 
 \\
 \rho_0  & =  & \varrho 
 \end{array}
  \right. 
\end{equation}
for all $A\in\mathfrak{L}\left(  \ell^2(\mathbb{Z}_+)\otimes \mathbb{C}^2\right)$.

Second, 
we will prove that (\ref{eq:8.4}) has at most one  $N^p $-weak solution
provided that the initial condition is  $N^p $-regular.
Suppose that (\ref{eq:8.8}) holds.
Taking $A=\left\vert y \rangle \langle x \right\vert $ in (\ref{eq:8.8}) we get 
 \begin{equation}
\label{eq:8.5}
 \frac{d}{dt} \left\langle x, \rho_{t} y \right\rangle 
= 
\left\langle  G \left( t \right)^{\ast}  x,  \rho_{t} y \right\rangle + \left\langle  x, \rho_{t}   G \left( t \right)^{\ast} y \right\rangle  
+ \sum_{\ell =1}^{3} \left\langle L_{\ell}^{\ast}   x, \rho_{t} L_{\ell}^{\ast}   y \right\rangle 
 \end{equation}
for all $ x,y \in  \mathcal{D}\left( N^p \right) $.
Relation (\ref{eq:8.5}) coincides with 
(\ref{eq:4.1}) with 
$\mathcal{A}_{t}$, $G \left( t \right)$, $L_{1}$, $L_{2}$ and $L_{3}$ replaced by 
$ \rho_{t} $, $G \left( t \right)^{\ast}$, $L_{1}^{\ast}$, $L_{2}^{\ast}$ and $L_{3}^{\ast}$.
This suggests us to apply Theorem \ref{th:UniquenessLinearHE} to (\ref{eq:8.5})
in order to prove the uniqueness of the solution of (\ref{eq:8.8}).
To this end,
we next deduce that the linear stochastic Schr\"odinger equation 
\begin{equation}
\label{eq:8.3}
Y_{t}\left( \xi \right) 
= \xi +\int_{0}^{t}G \left( s \right)^{\ast} Y_{s}\left( \xi \right) ds 
+ \sum_{\ell=1}^{3 }\int_{0}^{t}
L_\ell ^{\ast} \, Y_{s}\left( \xi \right) dW_{s}^{\ell} 
\end{equation}
satisfies Hypothesis \ref{hyp:CF-inequality} with $C=N^p$.

Now, we check  Hypothesis  \ref{hyp:CF-inequality} with 
$G \left( t \right)$, $L_{1}$, $L_{2}$ and $L_{3}$ replaced by 
$G \left( t \right)^{\ast}$, $L_{1}^{\ast}$, $L_{2}^{\ast}$ and $L_{3}^{\ast}$.
Take  $C=N^p$.
Since $a^\dagger $ is relatively bounded with respect to $N$,
using analysis similar to that in the second paragraph
we can check that 
$G \left( t \right)^{\ast} 
=
\mathrm{i} H \left( t \right) - \frac{1}{2} \sum_{\ell = 1}^3 L_{\ell}^{\ast} L_{\ell} $ 
satisfies Condition H4.1 of Hypothesis \ref{hyp:CF-inequality} with 
$G \left( t \right)$ substituted by $G \left( t \right)^{\ast}$,
as well as 
Condition  H4.2  holds with $L_{\ell} \left( t \right)$ replaced by 
$
L_1^{^{\ast}} =  \sqrt{ 2 \kappa} a^\dagger
$,
$
L_2 ^{\ast} =  \sqrt{ \gamma \left(1-d \right) } \sigma^{+}
$,
$ 
L_3  ^{\ast} =  \sqrt{ \gamma \left(1 +d \right) } \sigma^{-} 
$.
On $\mathfrak{D}$ we have 
\begin{eqnarray*}
G \left( t \right)^{\ast} + \left( G \left( t \right)^{\ast} \right)^{\ast} 
+
\sum_{\ell = 1}^3 \left( L_{\ell}^{\ast} \right)^{\ast}  L_{\ell}^{\ast}
& =  &
\sum_{\ell = 1}^3 \left(  L_{\ell}  L_{\ell}^{\ast} - L_{\ell}^{\ast} L_{\ell}  \right)
\\
 & = &
 4 \kappa ^2 I + 2 \gamma ^2 \left( 1 + d^2 \right) \sigma_3 , 
\end{eqnarray*}
which gives Condition H4.4.
For any $x \in \mathfrak{D}$,
\begin{eqnarray}
 \label{eq:8.1}
 && 2\Re\left\langle N^{2p} x, \mathrm{i} H \left( t \right) x\right\rangle 
 \\
 \nonumber
&& = 
 \sum_{k \in \mathbb{Z}_{+}, \eta = \pm} 
  2 \Re \left( \alpha \left( t \right)  x \left( k, \eta \right) \overline{ x \left(k+1, \eta \right) }  \right)
  \sqrt{k+1} \left(  \left( k+1 \right)^{2p} - k^{2p} \right)
 \end{eqnarray}
and
\begin{eqnarray}
 \label{eq:8.2}
&&  \left\langle x, 
   \left(  L_1  N^{2p}  L_1^{\ast} - \frac{1}{2} L_1^{\ast} L_1  N^{2p}   - \frac{1}{2} N^{2p}  L_1^{\ast} L_1  \right) x\right\rangle 
\\
\nonumber
&& =
  \sum_{k \in \mathbb{Z}_{+}, \eta = \pm}  2 \kappa \left\vert  x \left( k, \eta \right) \right\vert ^2 
  \left( \left( k +1 \right)^{2p+1} - k^{2p+1}  \right) .
 \end{eqnarray}
Since $L_2$, $L_3$ are bounded operators with conmute with $N^{2p}$,
using (\ref{eq:8.1}) and (\ref{eq:8.2}) yields 
$$
2\Re\left\langle N^{2p} x, G \left( t \right)^{\ast} x\right\rangle 
+ \sum_{\ell=1}^{3 }   \left\Vert N^{p} L_{\ell} ^{\ast} x \right\Vert ^{2}
\leq
K \left( t \right)  \left\Vert N^{p} x \right\Vert ^{2}
$$
and hence Condition  H4.3 holds.
By Theorem \ref{th:ExistUniqSSE},
(\ref{eq:8.3}) has a unique strong $N^p$-solution whenever  $\xi \in L_{C}^{2}\left( \mathbb{P}, \ell^2(\mathbb{Z}_+)\otimes \mathbb{C}^2\right) $.
It follows from Theorem \ref{th:UniquenessLinearHE} that (\ref{eq:8.5}) has at most one solution 
$ \varrho_t  \in \mathfrak{L} \left(  \ell^2(\mathbb{Z}_+)\otimes \mathbb{C}^2 \right) $
satisfying $ \varrho_0 = \varrho $.
Thus, (\ref{eq:8.4}) has a unique  $N^p$-regular solution, which is equal to
$
 \rho_t := \mathbb{E} \left|X_{t} \left( \xi \right)\right\rangle \left\langle X_{t} \left( \xi \right)\right| 
$.
\end{proof}

\subsubsection{Proof of Theorem \ref{th:Ehrenfest-Lineal}}
\label{sec:Proofth:Ehrenfest-Lineal}

 \begin{proof}
 From Theorem \ref{th:EyU-Lineal} it follows that 
 (\ref{eq:SSEp})  has a unique strong $N^p$-solution  $X_{t} \left( \xi \right)$
for any initial datum  $\xi \in L_{N^p}^{2}\left( \mathbb{P}, \ell^2(\mathbb{Z}_+)\otimes \mathbb{C}^2\right) $.
In order to establish (\ref{eq:8.10}) we apply Theorem 4.1 of  \cite{FagMora2013} to obtain 
\begin{eqnarray}
 \label{eq:8.9}
 \Tr \left(a \rho_t \right)
\hspace{-7pt} & = \hspace{-7pt} &
\Tr \left(a \rho_0 \right)
+
\sum_{\ell = 1}^{3} \int_0^t 
\mathbb{E} \left\langle  L_{\ell}  X_{s} \left( \xi \right) , a L_{\ell}  X_{s} \left( \xi \right) \right\rangle
ds
\\ \nonumber
& &
+
\int_0^t \left(
\mathbb{E} \left\langle a^\dagger  X_{s} \left( \xi \right) , G \left( s \right) X_{s} \left( \xi \right) \right\rangle
+ \mathbb{E} \left\langle G \left( s \right) X_{s} \left( \xi \right) , a  X_{s} \left( \xi \right) \right\rangle 
\right) ds ,
 \end{eqnarray}
where, throughout the proof, 
$G \left( t \right)$, $H \left( t \right)$, $L_1$, $L_2$, $L_3$ are as in Theorem \ref{th:EyU-Lineal}.
Therefore, 
$t \mapsto  \Tr \left(a \rho_t \right)$ is a continuous function.

Suppose that $x \in \mathfrak{D}$,
where $\mathfrak{D}$ is the set of all  $x \in  \ell^2(\mathbb{Z}_+)\otimes \mathbb{C}^2$ 
satisfying
$ \left\langle e_n \otimes e_{\eta}, x \right\rangle = 0$ 
for all combinations of  $n \in \mathbb{Z}_{+}$ and $\eta = \pm$ except a finite number. 
Since $a$ conmutes with $\sigma^{3}$ and $ \sigma^{\pm}$,
using $\left[ a , a^\dagger \right] = I$ we deduce that
\begin{eqnarray*}
 \left\langle a^\dagger  x , -  \mathrm{i} H \left( s \right)  x \right\rangle + \left\langle - \mathrm{i} H \left( s \right)  x , a x \right\rangle
& = &
\left\langle   x ,  \mathrm{i} \left[ H \left( s \right) , a  \right]  x \right\rangle
\\
& = &
\left\langle   x ,  \left[ 
\mathrm{i}  \omega \, a^\dagger a  
- \alpha \left( t \right) a^\dagger +  \overline{\alpha \left( t \right)} a 
 , a  \right]  x \right\rangle
\\
& = &
 \left\langle   x , 
 \left( - \mathrm{i} \omega \, a + \alpha \left( t \right)  \right) x \right\rangle 
\end{eqnarray*}
and
\begin{eqnarray*}
&& \sum_{\ell = 1}^{3}  \left\langle   x , \left(
L_{\ell}^{\star}  a L_{\ell} 
- \frac{1}{2} a L_{\ell}^{\star}  L_{\ell} -  \frac{1}{2}  L_{\ell}^{\star}  L_{\ell} a \right)
x \right\rangle
\\
&& = 
\left\langle   x , \left(
L_{1}^{\star}  a L_{1} 
- \frac{1}{2} a L_{1}^{\star}  L_{1} -  \frac{1}{2}  L_{1}^{\star}  L_{1} a \right)
x \right\rangle
 =
- \kappa \left\langle   x , a x \right\rangle .
\end{eqnarray*} 
Because $\mathfrak{D}$ is a core for $N$,
we obtain that for all $x \in \mathcal{D}\left( N \right)$,
$$
 \left\langle  a^\dagger   x ,  G \left( s \right) x \right\rangle
 +
 \left\langle G \left( s \right) x , a x \right\rangle 
 +
 \sum_{\ell = 1}^{3} \left\langle  L_{\ell} x , a L_{\ell}  x \right\rangle
 =
 \left\langle   x , - \left( \kappa   + \mathrm{i} \omega \right) a x + \alpha \left( t \right)   x \right\rangle .
$$
Then, from (\ref{eq:8.9}) it follows that
$$
\Tr \left(a \rho_t \right)
=
\Tr \left(a \rho_0 \right)
+
\int_0^t \left(
 - \left( \kappa   + \mathrm{i} \omega \right) \Tr \left(a \rho_s \right)
 +
  \alpha \left( s \right) 
\right) ds ,
$$
which leads to (\ref{eq:8.10}).

Fix $\eta = -$ or  $\eta = 3$.
According to (\ref{eq:8.8}) we have 
$$
 \frac{d}{dt}\Tr\left( \rho_{t} \sigma^{\eta} \right) 
 =
 \Tr\left(
\sigma^{\eta} \left( G \left( t \right) \rho_{t}\ +\rho_{t} G  \left( t \right)  ^{\ast}
+\sum_{\ell=1}^{3} L_{\ell}    \rho_{t} L_{\ell}  ^{\ast }\right) 
\right) ,
$$
and so applying Theorem 3.2 of \cite{MoraAP} we deduce that
\begin{eqnarray*}
&& \frac{d}{dt}\Tr\left( \rho_{t} \sigma^{\eta} \right) 
= 
\Tr\left( \rho_{t} 
 \left( \sigma^{\eta} G \left( t \right)  + G  \left( t \right)  ^{\ast} \sigma^{\eta}
+\sum_{\ell=1}^{3}   L_{\ell}  ^{\ast }  \sigma^{\eta} L_{\ell} \right)
\right) 
\\
&& = 
\Tr\left( \rho_{t}  \left(
-  \mathrm{i} \left[ \sigma^{\eta} , H \left( t \right) \right]
+ \sum_{\ell=1}^{3} \left(  L_{\ell}  ^{\ast }  \sigma^{\eta} L_{\ell} - \frac{1}{2}  \sigma^{\eta}  L_{\ell}  ^{\ast }  L_{\ell}
 - \frac{1}{2}   L_{\ell}  ^{\ast }  L_{\ell}  \sigma^{ \eta}\right)
\right) \right) 
\\
&&
 = 
\Tr\left( -  \mathrm{i}  \rho_{t} 
\left[ \sigma^{\eta} , 
\frac{\omega}{2}  \sigma^{3}+ \mathrm{i}  \left( \overline{\beta  \left( t \right) } \sigma^{-}  - \beta  \left( t \right) \sigma^+ \right)
 \right] \right) 
 \\
 && \quad
 + \sum_{\ell=2}^{3} \Tr\left( \rho_{t}  \left(
   \left(  L_{\ell}  ^{\ast }  \sigma^{\eta} L_{\ell} - \frac{1}{2}  \sigma^{\eta}  L_{\ell}  ^{\ast }  L_{\ell}
 - \frac{1}{2}   L_{\ell}  ^{\ast }  L_{\ell}  \sigma^{\eta}\right)
\right) \right) .
\end{eqnarray*}
Now, 
we use the commutation relations
$$
  \left[ \sigma^+, \sigma^-\right] = \sigma^3, \quad
   \left[\sigma^3, \sigma^+ \right] = 2 \sigma^+, \quad
    \left[\sigma^-, \sigma^3 \right] = 2 \sigma^- 
$$
to derive  (\ref{eq:8.11}) and (\ref{eq:8.12}).
\end{proof}

\subsubsection{Proof of Theorem \ref{th:LorenzEquations-Laser}}
\label{sec:ProofTh-LorenzEquations-Laser}

\begin{proof}
Fix $A \left( 0 \right)  \in \mathbb{C}$, $S \left( 0 \right)  \in \mathbb{C} $ and $D \left( 0 \right)  \in \mathbb{R}$.
Since (\ref{eq:Lorenz}) is an ordinary differential equation with locally Lipschitz coefficients,
(\ref{eq:Lorenz}) has a unique solution defined on a maximal interval $\left[ 0, T \right[$  
(see, e.g., \cite{Hirsch2013}).

For all $t \in  \left[ 0, T \right[$, we set
$X \left( t \right) =  \exp \left( i \omega t \right)  A \left( t \right)$, 
$ Y \left( t \right) =  \exp \left( i \omega t \right)  S \left( t \right) $
and 
$ Z \left( t \right)  =  D \left( t \right) - d $.
Thus, 
(\ref{eq:Lorenz}) becomes 
$$
 \left\{ 
\begin{array}{lcl}
 X ^{\prime} \left( t \right) 
 & = &
 - \kappa \, X \left( t \right)  + g \, Y \left( t \right)   
 \\
  Y ^{\prime} \left( t  \right)
&  =  &
d  g \, X \left( t  \right)  - \gamma \, Y \left( t  \right) + g \, X \left( t  \right) Z \left( t  \right)
 \\
 Z ^{\prime} \left( t  \right)
&  =  &
 -  4 g \,   \Re \left(
 \overline{ X  \left( t \right)}  \ Y \left( t  \right) 
 \right)
 - 2 \gamma  \ Z \left( t  \right)  
 \end{array}
  \right. .
$$
Therefore,
$$
 \frac{d}{dt} \left\vert X \left( t \right) \right\vert^2
 = 
 2 \, \Re \left( X^{\prime} \left( t \right)  \overline{ X \left( t \right) } \right) 
 = 
 -2 \kappa \left\vert X \left( t \right) \right\vert^2 + 2 g \, \Re \left( Y \left( t \right)  \overline{ X \left( t \right) } \right) 
$$
and 
$$
 \left\{
 \begin{array}{lcl}
 \frac{d}{dt} \left\vert Y \left( t \right) \right\vert^2
 & = &
  2 d g \, \Re \left( X \left( t \right)  \overline{ Y \left( t \right) } \right) 
 - 2 \gamma \left\vert  Y \left( t  \right)  \right\vert ^2
 + 2 g \, Z \left( t \right)  \Re \left( X \left( t \right)  \overline{ Y \left( t \right) } \right)
\\
\frac{d}{dt} Z \left( t \right)^2
& = &
- 4 \gamma  Z \left( t \right)^2 - 8 g \, Z \left( t \right)  \Re \left( \overline{ X \left( t \right) } Y \left( t \right)  \right)
 \end{array}
 \right. .
$$
Hence, 
\begin{equation}
\label{eq:L2}
4 \, \frac{d}{dt} \left\vert Y \left( t \right) \right\vert^2 + \frac{d}{dt}  Z \left( t \right)^2
=
 8 d g \, \Re \left( X \left( t \right)  \overline{ Y \left( t \right) } \right) 
- 8 \gamma \left\vert  Y \left( t  \right)  \right\vert ^2
- 4 \gamma Z \left( t \right)^2 .
\end{equation}

Suppose, for a moment, that $d < 0$.
Then 
$$
- 4 d \, \frac{d}{dt} \left\vert X \left( t \right) \right\vert^2
+ 4 \, \frac{d}{dt} \left\vert Y \left( t \right) \right\vert^2 + \frac{d}{dt} Z \left( t \right)^2
=
 8 d  \kappa \, \left\vert X \left( t \right) \right\vert^2 
 - 8 \gamma \left\vert  Y \left( t  \right)  \right\vert ^2
- 4 \gamma  Z \left( t \right)^2 .
$$
This gives
\begin{eqnarray*}
&& \frac{d}{dt} \left( - 4 d  \, \left\vert X \left( t \right) \right\vert^2 + 4 \,  \left\vert Y \left( t \right) \right\vert^2 
+  \left( Z \left( t \right) \right)^2 \right)
\\
&& \leq
- \min \left\{  2 \kappa, 2 \gamma \right\}
\left( - 4 d  \, \left\vert X \left( t \right) \right\vert^2 + 4 \,  \left\vert Y \left( t \right) \right\vert^2 
+   Z \left( t \right)^2 \right) ,
\end{eqnarray*}
which implies 
 \begin{eqnarray}
 \label{eq:L5a}
&& 4 \left\vert d \right\vert \, \left\vert X \left( t \right) \right\vert^2 + 4 \,  \left\vert Y \left( t \right) \right\vert^2 
+  Z \left( t \right)^2
\\
\nonumber
&& \leq
\exp \left(  - 2 t \, \min \left\{  \kappa, \gamma \right\} \right)
\left(  
4 \left\vert d \right\vert \, \left\vert X \left( 0 \right) \right\vert^2 + 4 \,  \left\vert Y \left( 0 \right) \right\vert^2 
+   Z \left( 0 \right)^2
\right)
 \end{eqnarray}
for any $t \in  \left[ 0, T \right[$.

On the other hand,
assume that $d \geq 0$.
Combining 
\begin{eqnarray*}
&& \frac{d}{dt}  \left\vert X \left( t \right) \right\vert^2
+
\frac{g^2}{4 \gamma \kappa} \left(  
4 \, \frac{d}{dt} \left\vert Y \left( t \right) \right\vert^2 + \frac{d}{dt}  Z \left( t \right)^2
\right)
\\
&& =
2 g \left( 1 + \frac{g^2 d}{\gamma \kappa} \right) \Re \left(   X  \left( t \right)  \overline{Y \left( t  \right)} \right) 
  - 2 \kappa \left\vert  X \left( t  \right)  \right\vert ^2
 - 2 \frac{g^2}{\kappa} \left\vert  Y \left( t  \right)  \right\vert ^2
 - \frac{g^2}{\kappa} Z \left( t \right) ^2 
\end{eqnarray*}
with
$
2 \Re \left(   X  \left( t \right)  \overline{\frac{g}{\kappa}  Y \left( t  \right)} \right) 
\leq
 \left\vert  X \left( t  \right)  \right\vert ^2 +  \frac{g^2}{\kappa^2} \left\vert  Y \left( t  \right)  \right\vert ^2
$
we obtain
\begin{eqnarray*}
&& \frac{d}{dt} \left(  \left\vert X \left( t \right) \right\vert^2
+ \frac{g^2}{\gamma \kappa} \left\vert Y \left( t \right) \right\vert^2 
+ \frac{g^2}{4 \gamma \kappa}  Z \left( t \right)^2 \right)
\\
&& \leq
\left( - \kappa + \frac{g^2 d}{\gamma}  \right)   \left\vert X \left( t \right) \right\vert^2
+ \left( - \gamma +   \frac{g^2 d}{\kappa} \right) \frac{g^2 }{\gamma  \kappa} \left\vert Y \left( t \right) \right\vert^2 
- 4 \gamma  \frac{g^2}{4 \gamma \kappa}  Z \left( t \right)^2 .
\end{eqnarray*}
Therefore, for all $t \in  \left[ 0, T \right[$ we have
\begin{eqnarray*}
&& \frac{d}{dt} \left(  \left\vert X \left( t \right) \right\vert^2
+ \frac{g^2}{\gamma \kappa} \left\vert Y \left( t \right) \right\vert^2 
+ \frac{g^2}{4 \gamma \kappa}  Z \left( t \right)^2 \right)
\\
&& \leq
- \min \left\{ \kappa - \frac{g^2 d}{\gamma} ,  \gamma  -  \frac{g^2 d}{\kappa} \right\}
\left(  \left\vert X \left( t \right) \right\vert^2
+ \frac{g^2}{\gamma \kappa} \left\vert Y \left( t \right) \right\vert^2 
+ \frac{g^2}{4 \gamma \kappa} Z \left( t \right)^2 \right) .
\end{eqnarray*}
This yields 
\begin{eqnarray}
  \label{eq:L4a}
&& \left\vert X \left( t \right) \right\vert^2
+ \frac{g^2}{\gamma \kappa} \left\vert Y \left( t \right) \right\vert^2 
+ \frac{g^2}{4 \gamma \kappa}  Z \left( t \right)^2
\\
\nonumber
&& \leq
\hbox{\rm e}^{ - t \min \left\{ \kappa - \frac{g^2 d}{\gamma} ,  \gamma  -  \frac{g^2 d}{\kappa} \right\} }
\left(  \left\vert X \left( 0 \right) \right\vert^2
+ \frac{g^2}{\gamma \kappa} \left\vert Y \left( 0 \right) \right\vert^2 
+ \frac{g^2}{4 \gamma \kappa} Z \left( 0 \right)^2 \right) .
 \end{eqnarray}

Suppose that $T < + \infty$.
According to (\ref{eq:L5a}) and (\ref{eq:L4a}) we have that  
$$
\left\Vert  \left( A \left( t \right), S \left( t \right), D \left( t \right)  \right) \right\Vert
< K ,
$$
where $K >0$ and $t \in \left[ 0, T \right[$.
This contradicts the property 
$$
\lim_{t \rightarrow T} \left\Vert  \left( A \left( t \right), S \left( t \right), D \left( t \right)  \right) \right\Vert
=
\infty .
$$
Therefore, $T = + \infty$.
Moreover, 
(\ref{eq:L5a}) and (\ref{eq:L4a})  lead to (\ref{eq:L5}) and (\ref{eq:L4}), respectively.
 \end{proof}

\subsection{Proof of Theorem \ref{th:EyU-LaserE}}
\label{sec:Proof:EyU-LaserE}

\begin{proof}
Let 
$\left( A \left( t \right), S \left( t \right), D \left( t \right)  \right) $
be the unique global solution of (\ref{eq:Lorenz}) with 
$  A \left( 0 \right) =  \Tr\left( a \varrho \right) $, $ S \left( 0 \right) = \Tr\left( \sigma^{-}  \varrho \right)$ 
and $ D \left( 0 \right) = \Tr\left(  \sigma^{3}  \varrho \right) $.
According to Theorem \ref{th:EyU-Lineal} we have that there exists a unique 
$N^p$-weak solution $\left( \rho_t \right)_{t \geq 0}$ to (\ref{eq:8.4})
with 
$ \alpha \left( t \right) =  g \, S \left( t \right)  $, $  \beta \left( t \right) =   g \, A \left( t \right) $
and initial datum $\rho_0 = \varrho$.
Moreover, 
Theorem \ref{th:EyU-Lineal} ensures that 
$\rho_t = \mathbb{E} \left| Z_{t} \left( \xi \right)\right\rangle \left\langle Z_{t} \left( \xi \right)\right| $,
where
$Z_{t} \left( \xi \right)$ is the strong $N^p$-solution of (\ref{eq:SSEp}) 
with 
$ \alpha \left( t \right) =  g \, S \left( t \right)  $,
$  \beta \left( t \right) =   g \, A \left( t \right) $
and
initial condition 
$\xi \in L_{N^p}^{2}\left( \mathbb{P}, \ell^2(\mathbb{Z}_+)\otimes \mathbb{C}^2\right) $ such that 
$
\varrho = \mathbb{E} \left| \xi \right\rangle \left\langle \xi \right| 
$.
Applying Theorem \ref{th:Ehrenfest-Lineal} we deduce that the evolutions of 
$\Tr\left( a \, \rho_{t} \right)$, $ \Tr\left( \sigma^{-} \rho_{t}   \right)$ and $ \Tr\left(  \sigma^{3}  \rho_{t} \right)$
are governed by  
\begin{equation}
 \label{eq:laser3}
 \left\{ 
 \begin{array}{lcl}
  \frac{d}{dt} \Tr\left( a \, \rho_{t} \right)
 & = &
 - \left( \kappa + \mathrm{i} \omega \right)  \Tr\left( a \, \rho_{t}  \right) + g \ S \left( t \right)  
 \\
  \frac{d}{dt} \Tr\left( \sigma^{-} \rho_{t}   \right)
&  =  &
 - \left( \gamma + \mathrm{i} \omega \right)   \Tr\left( \sigma^{-} \rho_{t}   \right) 
 + g \ A \left( t \right)  \Tr\left( \sigma^{3} \rho_{t}   \right) 
\\
 \frac{d}{dt} \Tr\left(  \sigma^{3}  \rho_{t} \right)
 & = &
 - 4 g \ \Re \left( \overline{A \left( t \right)}  \  \Tr\left(  \sigma^{-}  \rho_{t} \right)  
 \right)
 - 2 \gamma \left(  \Tr\left(  \sigma^{3} \rho_{t}  \right) -d \right) 
 \end{array}
   \right. .
\end{equation}
From the uniqueness of solution to (\ref{eq:laser3}) we find
$\Tr\left( a \, \rho_{t} \right) = A \left( t \right)$, $ \Tr\left( \sigma^{-} \rho_{t}   \right) = S \left( t \right)$ 
and $ \Tr\left(  \sigma^{3}  \rho_{t} \right) = D \left( t \right)$.
Hence
\begin{equation}
 \label{eq:Laser5}
 \left\{ 
  \begin{array}{lcl}
  \frac{d}{dt}\Tr\left( A \rho_{t}  \right) 
 & =  &
\Tr\left( A  \mathcal{L}_{\star} \left( \rho_t \right)  \rho_t \right) 
\quad \quad \quad \quad
\forall A\in\mathfrak{L}\left(  \ell^2(\mathbb{Z}_+)\otimes \mathbb{C}^2\right)
 \\
 \rho_0  & =  &  \varrho 
\end{array}
  \right. ,
\end{equation}
as well as
$ \alpha \left( t \right) =  g \, \mathbb{E} \left\langle Z_{t} \left( \xi \right), \sigma^{-} Z_{t} \left( \xi \right) \right\rangle $
and
$ \beta \left( t \right) =   g \,  \mathbb{E} \left\langle Z_{t} \left( \xi \right), a \, Z_t \left( \xi \right) \right\rangle$
(see, e.g., \cite{MoraAP}).
Therefore,
$Z_{t} \left( \xi \right)$ is a  strong $N^p$-solution of (\ref{eq:SSENonlinear}).

Let $Z_{t} \left( \xi \right)$ and $\widetilde{Z}_{t} \left( \xi \right)$ be  strong $N^p$-solutions of (\ref{eq:SSENonlinear})
with initial datum $\xi$ belonging to $ L_{N^p}^{2}\left( \mathbb{P}, \ell^2(\mathbb{Z}_+)\otimes \mathbb{C}^2\right) $.
Then, 
$Z_{t} \left( \xi \right)$ is the strong $N^p$-solution of (\ref{eq:SSEp})
with initial datum $\xi \in L_{N^p}^{2}\left( \mathbb{P}, \ell^2(\mathbb{Z}_+)\otimes \mathbb{C}^2\right) $,
$ \alpha \left( t \right) =  g \, \mathbb{E} \left\langle Z_{t} \left( \xi \right), \sigma^{-} Z_{t} \left( \xi \right) \right\rangle $
and
$ \beta \left( t \right) =   g \,  \mathbb{E} \left\langle Z_{t} \left( \xi \right), a \, Z_t \left( \xi \right) \right\rangle$.
Since 
$t \mapsto \mathbb{E} \left\langle Z_{t} \left( \xi \right), \sigma^{-} Z_{t} \left( \xi \right) \right\rangle $
and
$ t \mapsto \mathbb{E} \left\langle Z_{t} \left( \xi \right), a \, Z_t \left( \xi \right) \right\rangle$
are continuous functions,
applying Theorems \ref{th:EyU-Lineal} and \ref{th:Ehrenfest-Lineal},
together with Theorem 3.2 of \cite{MoraAP},
we deduce that 
$$
\mathbb{E} \left\langle Z_{t} \left( \xi \right), \sigma^{-} Z_{t} \left( \xi \right) \right\rangle ,
\quad
\mathbb{E} \left\langle Z_{t} \left( \xi \right), a \, Z_t \left( \xi \right) \right\rangle ,
\quad
\mathbb{E} \left\langle Z_{t} \left( \xi \right), \sigma^{3} Z_{t} \left( \xi \right) \right\rangle 
$$
is a solution of (\ref{eq:Lorenz}) with initial condition 
$  A \left( 0 \right) =  \Tr\left( a \, \varrho \right) $, $ S \left( 0 \right) = \Tr\left( \sigma^{-}  \varrho \right)$ 
and $ D \left( 0 \right) = \Tr\left(  \sigma^{3}  \varrho \right) $.
The same is true for $\widetilde{Z}_{t} \left( \xi \right)$ in place of $ Z_{t} \left( \xi \right)$,
and so Theorem \ref{th:LorenzEquations-Laser} leads to 
$
\mathbb{E} \left\langle Z_{t} \left( \xi \right), \sigma^{-} Z_{t} \left( \xi \right) \right\rangle 
=
\mathbb{E} \left\langle \widetilde{Z}_{t} \left( \xi \right), \sigma^{-} \widetilde{Z}_{t} \left( \xi \right) \right\rangle
$
and
$$
\mathbb{E} \left\langle Z_{t} \left( \xi \right), a \, Z_t \left( \xi \right) \right\rangle
=
\mathbb{E} \left\langle \widetilde{Z}_{t} \left( \xi \right), a \, \widetilde{Z}_t \left( \xi \right) \right\rangle
$$
for all $t \geq 0$.
Now,
the uniqueness of the strong $N^p$-solution of (\ref{eq:SSEp}) implies $Z = \widetilde{Z}$.

On the other hand,
suppose that 
$\left( \rho_t \right)_{t \geq 0}$ and $\left(\widetilde{\rho}_t \right)_{t \geq 0}$ 
are families of $N^p$-regular operators satisfying (\ref{eq:Laser5}) 
such that 
$ \rho_{0} = \widetilde{ \rho}_{0}  = \varrho$
and 
$t \mapsto \Tr\left( a \, \rho_{t}  \right) $, $t \mapsto \Tr\left( a \, \widetilde{\rho}_{t}  \right) $ are continuous.
Then, 
$\left( \rho_t \right)_{t \geq 0}$ is a $N^p$-weak solution to (\ref{eq:8.4})
with $ \alpha \left( t \right) =  g \, \Tr\left( \sigma^{-} \rho_{t}   \right) $ 
and $  \beta \left( t \right) =  g \, \Tr\left( a \, \rho_{t} \right) $,
as well as 
$\left(\widetilde{\rho}_t \right)_{t \geq 0}$ is a  $N^p$-weak solution to (\ref{eq:8.4})
with 
$ \alpha \left( t \right) =  g \, \Tr\left( \sigma^{-} \widetilde{\rho}_{t}  \right) $  
and $  \beta \left( t \right) = g \,  \Tr\left( a \widetilde{\rho}_{t} \right) $.
Using Theorem \ref{th:Ehrenfest-Lineal} 
we get that
$\left( \Tr\left( a \, \rho_{t} \right),  \Tr\left( \sigma^{-} \rho_{t}   \right),  \Tr\left(  \sigma^{3}  \rho_{t} \right) \right)$
and 
$$
\left( \Tr\left( a \, \widetilde{\rho}_{t} \right),  \Tr\left( \sigma^{-} \widetilde{\rho}_{t}   \right),  \Tr\left(  \sigma^{3} \widetilde{\rho}_{t} \right) \right)
$$
are solutions of (\ref{eq:Lorenz}) with initial condition 
$  A \left( 0 \right) =  \Tr\left( a \, \varrho \right) $, $ S \left( 0 \right) = \Tr\left( \sigma^{-}  \varrho \right)$ 
and $ D \left( 0 \right) = \Tr\left(  \sigma^{3}  \varrho \right) $.
Since the solution of (\ref{eq:Lorenz}) is unique (see, e.g., Theorem \ref{th:LorenzEquations-Laser}), 
$\Tr\left( a \, \rho_{t} \right) = \Tr\left( a \, \widetilde{\rho}_{t} \right)$,
$ \Tr\left( \sigma^{-} \rho_{t}   \right) = \Tr\left( \sigma^{-} \widetilde{\rho}_{t}   \right)$
and
$\Tr\left(  \sigma^{3}  \rho_{t} \right)  = \Tr\left(  \sigma^{3} \widetilde{\rho}_{t} \right) $.
Therefore,
$\left( \rho_t \right)_{t \geq 0}$ and $\left(\widetilde{\rho}_t \right)_{t \geq 0}$
are  $N^p$-weak solution to (\ref{eq:8.4})
with the same $ \alpha \left( t \right)$ and $ \beta \left( t \right) $,
and hence 
using  Theorem \ref{th:EyU-Lineal} yields $\rho_t = \widetilde{\rho}_t$ for all $t \geq 0$. 
\end{proof}

\end{document}